\documentclass[runningheads]{llncs}
\usepackage[T1]{fontenc}
\usepackage{amsmath,amssymb}
\usepackage{algorithmicx, algorithm,algpseudocode}
\usepackage{longtable}
\usepackage{tabularx}
\usepackage{enumerate}
\usepackage{graphicx}
\usepackage{algpseudocode}
\usepackage{algorithm}
\usepackage[title]{appendix}
\usepackage{stmaryrd}
\usepackage{float}
\usepackage{cite}
\usepackage{xcolor}
\usepackage{threeparttable} 

\newtheorem{fact}{\bf Fact}
\DeclareMathOperator{\Diag}{Diag} 
\usepackage{enumerate}
\usepackage{hyperref}
\hypersetup{
	colorlinks   = true, 
	urlcolor     = red, 
	linkcolor    = blue, 
	citecolor   = blue 
}

\newcommand{\splitatcommas}[1]{%
	\begingroup
	\begingroup\lccode`~=`, \lowercase{\endgroup
		\edef~{\mathchar\the\mathcode`, \penalty0 \noexpand\hspace{0pt plus 1em}}%
	}\mathcode`,="8000 #1%
	\endgroup
}
\usepackage[
n,
operators,
advantage,
sets,
adversary,
landau,
probability,
notions,	
logic,
ff,
mm,
primitives,
events,
complexity,
asymptotics,
keys]{cryptocode}

\begin{document}
\title{New Insights into Involutory and Orthogonal MDS Matrices}

\author{
    Yogesh Kumar\inst{1} \and
    Susanta Samanta\inst{2} \and
    Atul Gaur\inst{3}
}



\institute{
    Scientific Analysis Group, DRDO, Metcalfe House Complex, Delhi-110054 \\ \email{adhana.yogesh@gmail.com} 
    \and Department of Electrical and Computer Engineering, University of Waterloo, Waterloo, ON N2L3G1, Canada \\ \email{ssamanta@uwaterloo.ca} \and
    Department of Mathematics, University of Delhi, Delhi-110007 \\ \email{gaursatul@gmail.com}
}
\maketitle              
	

\begin{abstract}
MDS matrices play a critical role in the design of diffusion layers for block ciphers and hash functions due to their optimal branch number. Involutory and orthogonal MDS matrices offer additional benefits by allowing identical or nearly identical circuitry for both encryption and decryption, leading to equivalent implementation costs for both processes. These properties have been further generalized through the notions of semi-involutory and semi-orthogonal matrices. While much of the existing literature focuses on identifying efficiently implementable MDS candidates or proposing new constructions for MDS matrices of various orders, this work takes a different direction. Rather than introducing novel constructions or prioritizing implementation efficiency, we investigate structural relationships between the generalized variants and their conventional counterparts. Specifically, we establish nontrivial interconnections between semi-involutory and involutory matrices, as well as between semi-orthogonal and orthogonal matrices. Exploiting these relationships, we show that the number of semi-involutory MDS matrices can be directly derived from the number of involutory MDS matrices, and vice versa. A similar correspondence holds for semi-orthogonal and orthogonal MDS matrices. We also examine the intersection of these classes and show that the number of $3 \times 3$ MDS matrices that are both semi-involutory and semi-orthogonal coincides with the number of semi-involutory MDS matrices over~$\mathbb{F}_{2^m}$. Furthermore, we derive the general structure of orthogonal matrices of arbitrary order $n$ over $\mathbb{F}_{2^m}$. Finally, leveraging the aforementioned interconnections, we present an alternative and direct derivation of the explicit formulae for counting $3 \times 3$ semi-involutory MDS matrices and $3 \times 3$ semi-orthogonal MDS matrices.


\end{abstract}

\section{Introduction}
The concepts of confusion and diffusion~\cite{shan} play a crucial role in the design of symmetric key cryptographic primitives. Confusion aims to create a complex statistical relationship between the ciphertext and the plaintext, making it difficult for an attacker to exploit. Typically, confusion is achieved through the interaction of non-linear S-boxes with mixing and shuffling processes over multiple rounds. Diffusion, on the other hand, ensures that a change in a single bit of the plaintext results in a significant change in approximately half of the bits in the ciphertext. Similarly, altering a single bit of the ciphertext should cause about half of the bits in the plaintext to change. This is related to the expectation that encryption schemes exhibit an avalanche effect~\cite{avalanche_effect}. In many block ciphers and hash functions, the diffusion property is achieved through the use of a linear layer, represented as a matrix. This matrix is designed to significantly alter the output in response to a small change in the input. Hence, MDS matrices find significant applications in the design of block ciphers and hash functions.


MDS matrices with specific properties, such as being involutory or orthogonal, are of particular importance. An involutory matrix is one that is its own inverse, while an orthogonal matrix has its transpose as its inverse. These characteristics facilitate implementation by enabling the use of identical or nearly identical circuitry for both encryption and decryption processes. 

Researchers have consequently devoted considerable attention to involutory and orthogonal MDS matrices constructed from various canonical matrix families. Notably, Cauchy, Vandermonde, circulant, and Hadamard structures, as well as their generalizations, serve as rich sources for such matrices. The literature features many prominent examples, including those discussed in references~\cite{Gupta2023direct,kc2,GR15,LACAN2003,CYCLICM,Roth1989,psa,sdm,skf,Orthogonal_MDS2017}.

In addition to algebraic properties, several studies also focus on practical issues such as hardware cost and delay in the design and use of MDS matrices~\cite{Duval2018,Li2019FSE,IterativeMDS2019,Sajadieh2021,xiang2020}. To further expand these foundational ideas, the concepts of semi-involutory~\cite{cheon2021semi} and semi-orthogonal~\cite{FIEDLER2012} matrices have been introduced. The work in~\cite{chatterjee2023note} represents the first systematic exploration of these generalized structures for the construction of semi-involutory and semi-orthogonal MDS matrices. Follow-up studies~\cite{chatterjee2024_circulant,chatterjee2024_3_SIMDS,kumar2026study} have continued to advance this line of research. 

More specifically, \cite{chatterjee2023note} shows that certain Cauchy and Vandermonde based MDS constructions also satisfy the semi-orthogonal property, and it initiates a study of circulant matrices with semi-involutory and semi-orthogonal structure. In~\cite{chatterjee2024_circulant}, the authors further analyze circulant matrices over finite fields and relate the MDS property to the trace of an associated diagonal matrix. In~\cite{chatterjee2024_3_SIMDS}, they characterize $3\times 3$ semi-involutory MDS matrices and enumerate them over finite fields of characteristic 2. Independently, \cite{kumar2026study} derives counting formula for $3\times 3$ semi-involutory MDS matrices over $\mathbb{F}_{2^m}$, and additionally provides algorithms to test the semi-involutory and semi-orthogonal properties, detailed structural results for circulant matrices with these properties, and non-existence results for circulant MDS matrices under the same structural constraints.

\vspace{1em}
\noindent\textbf{Our Contribution.} 
While much of the existing literature focuses on efficiently implementable MDS candidates or proposes new constructions for MDS matrices of various orders, this work adopts a different perspective by investigating structural relationships. In particular, we study the connections between generalized matrix forms, such as semi-involutory and semi-orthogonal matrices, and their conventional counterparts: involutory and orthogonal matrices. These relationships reveal deeper algebraic insights and provide a framework for systematically determining the number of MDS matrices that satisfy additional structural constraints.

More specifically, our key contributions are summarized as follows:
\begin{itemize}
    \setlength{\itemsep}{1em}
    \item[$\bullet$] We establish nontrivial connections between semi-involutory and involutory matrices, and between semi-orthogonal and orthogonal matrices over $\mathbb{F}_{2^m}$. These relationships enable bidirectional derivation: the number of semi-involutory MDS matrices over $\mathbb{F}_{2^m}$ determines the number of involutory MDS matrices over $\mathbb{F}_{2^m}$ and vice versa, with the same correspondence holding for semi-orthogonal and orthogonal MDS matrices.

    \item[$\bullet$] We characterize matrices that simultaneously satisfy both the semi-involutory and semi-orthogonal properties. Using this characterization, we prove that the number of $3 \times 3$ MDS matrices exhibiting both properties is equal to the number of $3 \times 3$ semi-involutory MDS matrices over $\mathbb{F}_{2^m}$.
    
    \item[$\bullet$] We derive the general structure of orthogonal matrices of arbitrary order $n$ over $\mathbb{F}_{2^m}$. Specializing this result to the case $n = 3$, we prove that the number of $3 \times 3$ orthogonal MDS matrices over $\mathbb{F}_{2^m}$ is $(2^m - 2)(2^m - 3)(2^m - 4)$.  

    \item[$\bullet$] We revisit the enumeration of $3 \times 3$ semi-involutory and semi-orthogonal MDS matrices over $\mathbb{F}_{2^m}$. Exploiting their correspondence with involutory MDS matrices and known enumeration results from~\cite{gmt}, we provide a direct derivation of the closed-form count for $3\times 3$ semi-involutory MDS matrices over $\mathbb{F}_{2^m}$. In an analogous way, we derive the closed-form count for $3\times 3$ semi-orthogonal MDS matrices over $\mathbb{F}_{2^m}$. Together with previously known formulas and those proved in this work, these results yield a consolidated set of explicit enumeration formulas for several algebraically constrained $3\times 3$ MDS classes, summarized in Table~\ref{tab:MDS_enumeration}.
    
    \item[$\bullet$] We extend our enumeration to $4 \times 4$ matrices by computing the number of semi-involutory MDS matrices over $\mathbb{F}_{2^m}$ for $m = 3, 4, \ldots, 8$, thereby extending the previously known results~\cite{kumar2026study} for $m=3,4$ to larger finite field. Finally, Tables~\ref{Table:Count_Comparison_3_MDS} and~\ref{Table:Count_Comparison_4_MDS} present a comparison of the counts of $3 \times 3$ and $4 \times 4$ MDS matrices with various structural properties.

\end{itemize}

\begin{table}[h]
    \centering
    \caption{Enumeration formulas for $3\times 3$ MDS matrices over $\mathbb{F}_{2^m}$}
    \label{tab:MDS_enumeration}
	\vspace{2mm}
    \scalebox{0.9}{%
	\begin{tabular}{|>{\raggedright\arraybackslash}p{4cm}|l|}
        \hline
       {Matrix Type} & {Number of Matrices} \\
	    \hline
	    $3\times 3$ MDS & $(2^m-1)^5 (2^m-2)(2^m-3)(2^{2m}-9\cdot 2^m+21)$ \cite{Kumar_MDS2024} \\[4pt]
	    $3\times 3$ Involutory MDS & $(2^m-1)^2 (2^m-2)(2^m-4)$ \cite{gmt}\\[4pt]
	    $3\times 3$ Semi-involutory MDS & $(2^m-1)^5 (2^m-2)(2^m-4)$ \cite{chatterjee2024_3_SIMDS,kumar2026study} [This work, Corollary~\ref{Th_count_SIMDS_3}]\\[4pt]
	    $3\times 3$ Orthogonal MDS & $(2^m-2)(2^m-3)(2^m-4)$ [This work, Theorem~\ref{count_ortho_mds}] \\[4pt]
	    $3\times 3$ Semi-orthogonal MDS & $(2^m-1)^5(2^m-2)(2^m-3)(2^m-4)$ \cite{kumar2026study} [This work, Theorem~\ref{Th_count_SOMDS_3}]\\[4pt]
	    $3\times 3$ MDS with both Semi-involutory and Semi-orthogonal & $(2^m-1)^5(2^m-2)(2^m-4)$ [This work, Theorem~\ref{count_3_SISO-MDS}] \\
	    \hline            
	\end{tabular}%
    }
\end{table}

\vspace*{1em}

\noindent The structure of the paper is as follows. In Section~\ref{Sec:Definition}, we provide a brief discussion of the mathematical background and notations employed throughout the paper. Section~\ref{Sec:Interconnections} establishes the interconnections between semi-involutory and involutory matrices, as well as between semi-orthogonal and orthogonal matrices. Section~\ref{Sec:orthogonal_count} analyzes the general structure of orthogonal matrices of arbitrary order $n$ over $\mathbb{F}_{2^m}$ and presents a comparison of the counts of $3 \times 3$ and $4 \times 4$ MDS matrices with various structural properties. Finally, Section~\ref{Sec:Conclusion} concludes the paper.

\section{Mathematical Preliminaries}~\label{Sec:Definition}

In this section, we discuss some definitions and mathematical preliminaries that are important in our context.

Let $\mathbb{F}_{2^m}$ be the finite field of order $2^m$. We denote the multiplicative group of the finite field $\mathbb{F}_{2^m}$ by $\mathbb{F}_{2^m}^*$. The set of vectors of length $n$ with entries from the finite field $\mathbb{F}_{2^m}$ is denoted by $\mathbb{F}_{2^m}^n$.

Let $M$ be any $n\times n$ matrix over $\mathbb{F}_{2^m}$ and let $|M|$ denote its determinant. An $n \times n$ matrix is referred to as a matrix of order $n$. A matrix $D$ of order $n$ is diagonal if $(D){ij}=0$ for $i\neq j$. Using $d_i = (D){ii}$, we represent the diagonal matrix $D$ as $\Diag(d_1, d_2, \ldots, d_n)$. The determinant of $D$ is $|D| = \prod_{i=1}^{n}{d_i}$. Hence, $D$ is non-singular over $\mathbb{F}_{2^m}$ if and only if $d_i \neq 0$ for $1 \leq i \leq n$.

MDS matrix finds its practical applications as a diffusion layer in cryptographic primitives. The concept of the MDS matrix comes from coding theory, specifically from the realm of maximum distance separable (MDS) codes. An $[n, k, d]$ code is MDS if it meets the singleton bound $d = n-k + 1$.

\begin{theorem}~\cite[page 321]{FJ77} 
    An $[n, k, d]$ code $C$ with a generator matrix $G = [ I ~|~ M ]$, where $M$ is a $k \times ( n - k )$ matrix, is MDS if and only if every square sub-matrix (formed from any $i$ rows and any $i$ columns, for any $i = 1, 2,\ldots, min \{k, n - k \}$) of $M$ is non-singular.
\end{theorem}

\begin{definition}
    A matrix $M$ of order $n$ is said to be an MDS matrix if $[I~|~M]$ is a generator matrix of an $[2n,n]$ MDS code.
\end{definition}

Another way to define an MDS matrix is as follows:

\begin{fact}
    A square matrix $M$ is an MDS matrix if and only if every square sub-matrix of $M$ is non-singular. 
\end{fact}

\noindent One of the elementary row operations on matrices is multiplying a row of a matrix by a non-zero scalar. MDS property remains invariant under such operations. Thus, we have the following result regarding MDS matrices.

\begin{lemma}\cite{kcz}\label{Lemma_DMD_MDS}
    Let $M$ be an MDS matrix, then for any non-singular diagonal matrices $D_1$ and $D_2$, $D_1MD_2$ will also be an MDS matrix.
\end{lemma}

\noindent Implementing involutory or orthogonal MDS matrices is more advantageous since it enables the use of an equivalent circuit or nearly equivalent circuit for both the encryption and decryption processes.

\begin{definition}
An involutory matrix is defined as a square matrix $M$ that is self-invertible or, equivalently, $M = M^{-1}$.
\end{definition}

\begin{definition}
	An orthogonal matrix is defined as a square matrix $M$ whose inverse is its transpose or, equivalently, $M^{-1} = M^T$.
\end{definition}

In 2021, Cheon et al.~\cite{cheon2021semi} introduced the concept of semi-involutory property as a generalization of involutory property. The definition is provided below.

\begin{definition}\cite{cheon2021semi}\label{Def_SI}
    A non-singular matrix $M$ is said to be semi-involutory if there exist two non-singular diagonal matrices $D$ and $D'$ such that $M^{-1} = DMD'$.
\end{definition}

We refer to the diagonal matrices $D$ and $D'$ as the corresponding diagonal matrices of the semi-involutory matrix $M$. 



\noindent Furthermore, in 2012, Fiedler et al.~\cite{FIEDLER2012} generalized orthogonal matrices and termed them G-matrices, which we refer to as semi-orthogonal matrices throughout this paper. The definition of a semi-orthogonal matrix is as follows:

\begin{definition}~\cite{FIEDLER2012} 
A non-singular matrix $M$ is semi-orthogonal if there exist two non-singular diagonal matrices $D$ and $D'$ such that $M^{-T} = DMD'$, where $M^{-T}$ denotes the transpose of the matrix $M^{-1}$.
\end{definition}

Similar to semi-involutory matrices, we refer to the diagonal matrices $D$ and $D'$ as the corresponding diagonal matrices of the semi-orthogonal matrix $M$.


\begin{lemma}\label{lemma_D1MD2_SI_SO}
    For any two non-singular diagonal matrices $P$ and $Q$, if $M$ is semi-involutory (or semi-orthogonal), then $PMQ$ is also a semi-involutory (or semi-orthogonal) matrix.
\end{lemma}
\begin{proof}
    Let $P$ and $Q$ be two non-singular diagonal matrices. Suppose that $M$ is a semi-involutory matrix with corresponding diagonal matrices $D$ and $D'$, i.e., $M^{-1}=DMD'$. Now, we have
    \begin{eqnarray*}
    (PMQ)^{-1}&=&Q^{-1}M^{-1}P^{-1}\\
    &=&PP^{-1}Q^{-1}(DMD')P^{-1}QQ^{-1}\\
    &=&P^{-1}Q^{-1}D(PMQ)D'P^{-1}Q^{-1}.
    \end{eqnarray*}

    Since $P^{-1}Q^{-1}D$ and $D'P^{-1}Q^{-1}$ are also non-singular diagonal matrices, therefore, $PMQ$ is a semi-involutory matrix. Similarly, we can prove that if $M$ is a semi-orthogonal matrix, then $PMQ$ is also a semi-orthogonal matrix. \qed
\end{proof}


\section{The Interconnections}\label{Sec:Interconnections}
In~\cite{Kumar_MDS2024}, the authors introduce a technique for generating all $n \times n$ MDS and involutory MDS matrices over $\mathbb{F}_{2^m}$. The proposed method involves first identifying $n \times n$ representative MDS matrices using a search-based approach. Subsequently, all $n \times n$ MDS and involutory MDS matrices can be obtained by multiplying two diagonal matrices with these representative matrices. To find all $n \times n$ representative MDS matrices, the authors define the representative matrix form $M_1$ over $\mathbb{F}_{2^m}^*$ as follows:

\begin{equation}\label{Eqn_The_matrixM1}
	M_1=
	\begin{pmatrix}
		1&1&\ldots&1\\
		1& & &\\
		\vdots & &R&\\
		1&  & &
	\end{pmatrix},
\end{equation}
where $R$ is a $(n-1) \times (n-1)$ matrix.

They also present the unique decomposition of a matrix $M$ over $\mathbb{F}_{2^m}^*$ in the form $M=D_1M_1D_2$, where $D_1$ and $D_2$ are non-singular diagonal matrices. The result is stated below.

\begin{theorem}\cite{Kumar_MDS2024}\label{Th_decomposition}
    Let $M = (m_{ij})$ be a $n \times n$ matrix over $\mathbb{F}_{2^m}^*$. Then, there exist unique $n \times n$ matrices $D_1$, $D_2$, and $M_1$ over $\mathbb{F}_{2^m}^*$ such that 
    \[M = D_1 M_1 D_2,\] 
    where $D_1$ and $D_2$ are diagonal matrices, the first entry of $D_2$ is $1$, and $M_1$ is a matrix of the form given in (\ref{Eqn_The_matrixM1}).
\end{theorem}

We will denote this unique decomposition as $M = \Phi(D_1, D_2, M_1)$. Throughout this section, all results are established over $\mathbb{F}_{2^m}^*$. Since MDS matrices inherently cannot contain zero entries, their properties align naturally with $\mathbb{F}_{2^m}^*$. To streamline the presentation, we will refer to the results as being over $\mathbb{F}_{2^m}$ when discussing MDS matrices throughout the paper.

\subsection{The Interconnection between Involutory and Semi-involutory Matrices}

In this section, we utilize the representative matrix form to derive the interconnections between semi-involutory and involutory matrices. Specifically, in Theorem~\ref{Th_count_IMDS}, we establish that the number of $n\times n$ involutory MDS matrices over $\mathbb{F}_{2^m}$ is equal to $(2^m - 1)^{n-1}$ times the count of semi-involutory MDS matrices of order $n$ of the form $M_1$ over $\mathbb{F}_{2^m}$. To derive Theorem~\ref{Th_count_IMDS}, we first reference the theorem presented in \cite{Kumar_MDS2024}.

\begin{theorem}~\cite{Kumar_MDS2024}\label{Th_inv_semi_inv_rel}
    Let $M_1=(c_{ij})$ be a matrix of order $n$ over $\mathbb{F}_{2^m}^*$, as specified in~(\ref{Eqn_The_matrixM1}), and $M_2=(d_{ij})$ be its inverse. Then $\Phi(D_1,D_2,M_1)$ will be an involutory matrix if and only if $\exists~ \alpha_i \in \mathbb{F}_{2^m}^{*}$ such that
	\begin{center}
		$d_{ij}=\alpha_i \alpha_j c_{ij}$, $1\leq i,j\leq n$.
	\end{center}
	Moreover, $D_1$ and $D_2$ must take the following form:
	\begin{center}
		$D_1=\Diag(\alpha_1,\lambda_2,\lambda_3,\ldots,\lambda_n)$ and $D_2=\Diag(1,\frac{\alpha_2}{\lambda_2},\frac{\alpha_3}{\lambda_3},\ldots,\frac{\alpha_n}{\lambda_n})$,
	\end{center}
	where $\lambda_i \in \mathbb{F}_{2^m}^{*}$ for $i=2,3,\ldots,n$.
\end{theorem}

\begin{corollary}\label{Coro_SI}
	The number of involutory matrices of order $n$ over $\mathbb{F}_{2^m}^{*}$ is given by 
	\[
		(2^m - 1)^{n-1} \cdot N_1,
	\]
	where $N_1$ denotes the number of semi-involutory matrices of order $n$ over $\mathbb{F}_{2^m}^{*}$ of the form $M_1$, as specified in (\ref{Eqn_The_matrixM1}).
\end{corollary}

\begin{proof}
	Let $M$ be an involutory matrix of order $n$ over $\mathbb{F}_{2^m}^{*}$. Then by Theorem~\ref{Th_decomposition}, there exists a unique $M_1$ of the form of~(\ref{Eqn_The_matrixM1}) such that $M=D_1M_1D_2$ for some non-singular diagonal matrices $D_1$ and $D_2$. 
    Also, from Theorem~\ref{Th_inv_semi_inv_rel}, when $M$ is involutory, $D_1$ and $D_2$ have the particular form. We can see that  
	\begin{equation*}
		\begin{aligned}
			&M_1^{-1}=D_2D_1M_1D_2D_1.\\
		\end{aligned}
	\end{equation*}
    
    Since $D_2D_1$ is again a non-singular diagonal matrix, from Definition~\ref{Def_SI}, $M_1$ is a semi-involutory matrix with the corresponding diagonal matrices $D$ and $D'$, where $D = D'=D_2D_1$. 
        
	Moreover, it is possible to construct other involutory matrices by selecting different values for $D_1$ and $D_2$. However, according to Theorem~\ref{Th_inv_semi_inv_rel}, the matrix $\Phi(D_1, D_2, M_1)$ will be involutory only when $D_1$ and $D_2$ assume the specific form described. The scalar values $\lambda_i$ for $i=2,3,\ldots,n$ can be chosen as any non-zero element from the field $\mathbb{F}_{2^m}$. Thus, in total, we have $(2^m - 1)^{n-1}$ possible choices for $D_1$ and $D_2$. 
    
    Hence, for a given semi-involutory matrix $M_1$, we can construct $(2^m - 1)^{n-1}$ distinct involutory matrices. Since each matrix $M$ over $\mathbb{F}_{2^m}^*$ corresponds uniquely to an $M_1$, we conclude that the total number of involutory matrices over $\mathbb{F}_{2^m}^*$ is given by

    \begin{equation*}
        \begin{aligned}
            (2^m - 1)^{n-1} \cdot N_1,
        \end{aligned}
    \end{equation*}
    where $N_1$ represents the number of semi-involutory matrices of order $n$ over $\mathbb{F}_{2^m}^{*}$ of the form $M_1$. \qed
\end{proof}


Thus, from the above corollary, we can easily derive the following result for involutory MDS matrices over the finite field $\mathbb{F}_{2^m}$.

\begin{theorem}\label{Th_count_IMDS}
	The number of involutory MDS matrices of order $n$ over $\mathbb{F}_{2^m}$ is given by 
	\[
		(2^m - 1)^{n-1} \cdot N_2,
	\]
	where $N_2$ denotes the number of semi-involutory MDS matrices of order $n$ over $\mathbb{F}_{2^m}$ of the form $M_1$, as specified in (\ref{Eqn_The_matrixM1}).
\end{theorem}

\begin{proof}
	Let $M = \Phi(D_1, D_2, M_1)$ be an involutory MDS matrix of order $n$ over $\mathbb{F}_{2^m}$. Since $D_1$ and $D_2$ are non-singular matrices, by Lemma~\ref{Lemma_DMD_MDS}, $M_1$ is also an MDS matrix. Additionally, as $M$ is an involutory matrix, Theorem~\ref{Th_inv_semi_inv_rel} implies that $M_1$ is a semi-involutory matrix.
	
    Now, since each MDS matrix $M$ over $\mathbb{F}_{2^m}$ corresponds uniquely to an MDS matrix of the form $M_1$, we conclude from Corollary~\ref{Coro_SI}~\footnote{It is worth noting that the result in Corollary~\ref{Coro_SI} is stated for $\mathbb{F}_{2^m}^*$. However, since MDS matrices cannot contain zero entries, this condition is inherently satisfied by the results over $\mathbb{F}_{2^m}^*$. The same holds for the results in Theorems~\ref{Th_count_SIMDS}, \ref{Th_count_Orthogonal_MDS}, and \ref{Th_count_Semi-OMDS}, as they are derived from results over $\mathbb{F}_{2^m}^*$.} that the total number of involutory MDS matrices over $\mathbb{F}_{2^m}$ is given by
	\begin{equation*}
		\begin{aligned}
			(2^m - 1)^{n-1} \cdot N_2,
		\end{aligned}
	\end{equation*}
	where $N_2$ represents the number of semi-involutory MDS matrices of order $n$ over $\mathbb{F}_{2^m}$ of the form $M_1$.
    \qed
\end{proof}

Now, we will discuss the results related to the semi-involutory MDS matrices over $\mathbb{F}_{2^m}$. For this purpose, we need the following lemma.

\begin{lemma}\label{Lemma_complete_SI}
	The number of semi-involutory matrices of order $n$ over $\mathbb{F}_{2^m}^{*}$ is given by 
	\[
		(2^m - 1)^{2n-1} \cdot N_1,
	\]
	where $N_1$ denotes the number of semi-involutory matrices of order $n$ over $\mathbb{F}_{2^m}^{*}$ of the form $M_1$, as specified in (\ref{Eqn_The_matrixM1}).
\end{lemma}

\begin{proof}
	Let $M=\Phi(D_1,D_2,M_1)$ be a semi-involutory matrix over $\mathbb{F}_{2^m}^{*}$ with corresponding diagonal matrices $D$ and $D'$. Thus, we have 

	\begin{equation*}
		\begin{aligned}
		& M^{-1} =DMD'\\
		\implies & M (DMD')=I \\
		\implies & D_1M_1D_2(D(D_1M_1D_2) D')= I \\
		\implies & D_1M_1D_2(DD_1M_1D_2D')= I \\
		\implies & M_1(D_2DD_1M_1D_2D' D_1)=I.
		\end{aligned}
	\end{equation*}

	Thus, the inverse of the matrix $M_1$ is given by 
	\begin{equation*}
		\begin{aligned}
			&M_1^{-1}=D_2DD_1M_1D_2D' D_1.
		\end{aligned}
	\end{equation*}
	
	Now, since $D_2DD_1$ and $D_2D'D_1$ are non-singular diagonal matrices, we can say that $M_1$ is also a semi-involutory matrix.

	\noindent Conversely, suppose that $M_1$ is a semi-involutory matrix. Then, by Lemma~\ref{lemma_D1MD2_SI_SO}, we can say that $M=\Phi(D_1,D_2,M_1)=D_1M_1D_2$ is also a semi-involutory matrix.

	Therefore, we have established that $M=\Phi(D_1,D_2,M_1)$ is a semi-involutory matrix if and only if $M_1$ is a semi-involutory matrix.

	Now, by selecting different values for $D_1$ and $D_2$, it is possible to construct other semi-involutory matrices from $M_1$. According to Theorem~\ref{Th_decomposition}, we have $(2^m - 1)^{2n-1}$ possible choices for $D_1$ and $D_2$. Therefore, we conclude that the total number of semi-involutory matrices over $\mathbb{F}_{2^m}^*$ is given by
	\begin{equation*}
		\begin{aligned}
			(2^m - 1)^{2n-1} \cdot N_1,
		\end{aligned}
	\end{equation*}
	where $N_1$ represents the number of semi-involutory matrices of order $n$ over $\mathbb{F}_{2^m}^{*}$ of the form $M_1$.
    \qed
\end{proof}

\noindent Similar to Theorem~\ref{Th_count_IMDS}, using the above lemma, we can easily derive the result for semi-involutory MDS matrices over $\mathbb{F}_{2^m}$. Thus, we provide the result without proof.

\begin{theorem}\label{Th_count_SIMDS}
	The number of semi-involutory MDS matrices of order $n$ over $\mathbb{F}_{2^m}$ is given by
	\[
		(2^m - 1)^{2n-1} \cdot N_2,
	\]
	where $N_2$ denotes the number of semi-involutory MDS matrices of order $n$ over $\mathbb{F}_{2^m}$ of the form $M_1$, as specified in (\ref{Eqn_The_matrixM1}).
\end{theorem}

\noindent In the following lemma, we mention the explicit formula obtained in~\cite{gmt} for enumerating all $3\times 3$ involutory MDS matrices over $\mathbb{F}_{2^m}$.

\begin{lemma}\cite{gmt}\label{lemma:inv_3_MDS_count}
    For $m\geq3$, the number of all $3 \times 3$ involutory MDS matrices over $\mathbb{F}_{2^m}$ is given by $(2^m-1)^2 (2^m-2)(2^m-4)$.
\end{lemma}

\begin{corollary}\label{Th_count_SIMDS_3}
    For $m\geq3$, the number of semi-involutory MDS matrices of order $3$ over $\mathbb{F}_{2^m}$ is $(2^m-1)^5(2^m-2)(2^m-4)$.
\end{corollary}

\begin{proof}
    From the above lemma, we know that the number of all $3 \times 3$ involutory MDS matrices is $(2^m-1)^2(2^m-2)(2^m-4)$. Therefore, from Theorem~\ref{Th_count_IMDS}, we have
	\begin{equation*}
		\begin{aligned}
			&(2^m-1)^2 (2^m-2)(2^m-4)= (2^m-1)^2 \cdot N_2\\
			\implies & N_2= (2^m-2)(2^m-4).
		\end{aligned}
	\end{equation*}
	Here $N_2$ denotes the number of semi-involutory MDS matrices of order $3$ over $\mathbb{F}_{2^m}$ of the form $M_1$, as specified in (\ref{Eqn_The_matrixM1}).

	Therefore, from Theorem~\ref{Th_count_SIMDS}, the count of semi-involutory MDS matrices of order $3$ over $\mathbb{F}_{2^m}$ is given by
    \begin{equation*}
        \begin{aligned}
            &(2^m-1)^{2\cdot 3-1}\cdot N_2\\
            =& (2^m-1)^5(2^m-2)(2^m-4).
        \end{aligned}
    \end{equation*}
    \qed
\end{proof}


\begin{remark}
    We note that the enumeration formula for $3 \times 3$ semi-involutory MDS matrices over $\mathbb{F}_{2^m}$ has been previously established by Chatterjee et al.~\cite{chatterjee2024_3_SIMDS} through a comprehensive proof, and independently derived in recent work by Kumar et al.~\cite{kumar2026study}. Our approach complements these results by leveraging the structural interconnection with involutory MDS matrices, which yields an alternative derivation that is more direct and concise.
\end{remark}




\noindent In~\cite{Kumar_4MDS}, the authors propose a matrix form for generating all $4 \times 4$ involutory MDS matrices over the field $\mathbb{F}_{2^m}$. Their approach specifically involves searching for these representative matrices within a set of cardinality $(2^m-1)^5$, facilitating the explicit enumeration of the total number of $4 \times 4$ involutory MDS matrices over $\mathbb{F}_{2^m}$ for $m = 3, 4, \ldots, 8$. We take their results and combine Theorems~\ref{Th_count_IMDS} and \ref{Th_count_SIMDS} to derive the exact count of $4 \times 4$ semi-involutory MDS matrices over $\mathbb{F}_{2^m}$. For example, as derived in \cite{Kumar_4MDS}, the count of $4\times 4$ involutory MDS matrices over $\mathbb{F}_{2^3}$ is $7^3 \times 48$. By applying Theorem~\ref{Th_count_IMDS}, the number of semi-involutory MDS matrices of the form $M_1$ is $48$. Therefore, using Theorem~\ref{Th_count_SIMDS}, we conclude that the number of $4 \times 4$ semi-involutory MDS matrices over $\mathbb{F}_{2^3}$ is $7^7 \times 48$. Similarly, we can derive the count for the finite fields $\mathbb{F}_{2^m}$, where $m = 4, \dots, 8$. A comparison of the counts for $4 \times 4$ involutory and semi-involutory MDS matrices is presented in Table~\ref{Table:count_4_SIMDS}.

\begin{table}[htbp]
    \centering
	\caption{Count of $4\times 4$ involutory and semi-involutory MDS Matrices}
	\label{Table:count_4_SIMDS}
	\vspace{2mm}
	\begin{tabular}{|c|c|c|c|}\hline
		Finite Field & Semi-involutory $M_1$ & Involutory & Semi-involutory  \\ \hline
		$\mathbb{F}_{2^3}$ & $48$       & $7^3\times48$  &$7^7\times48$\\   
		$\mathbb{F}_{2^4}$ & $71856$    & $15^3\times71856$ &$15^7\times71856$\\ 
		$\mathbb{F}_{2^5}$ & $10188240$    & $31^3\times10188240$ &$31^7\times10188240$\\ 
		$\mathbb{F}_{2^6}$ & $612203760$    & $63^3\times612203760$ &$63^7\times612203760$\\ 
		$\mathbb{F}_{2^7}$ & $26149708368$    &$127^3\times26149708368$ &$127^7\times26149708368$\\ 
		$\mathbb{F}_{2^8}$ & $961006331376$    &$255^3\times961006331376$ &$255^7\times961006331376$\\ 
		\hline
	\end{tabular}
	\vspace{2mm}
\end{table}

\subsection{The Interconnection between Orthogonal and Semi-orthogonal Matrices}

In this section, we utilize the representative matrix form $M_1$, as defined in~(\ref{Eqn_The_matrixM1}), to explore the interconnections between semi-orthogonal and orthogonal matrices. Specifically, in Theorem~\ref{Th_count_Semi-OMDS}, we demonstrate that the number of orthogonal MDS matrices of order $n$ over $\mathbb{F}_{2^m}$ corresponds to the count of semi-orthogonal MDS matrices of order $n$ of the form $M_1$. To establish this relationship, we present the following theorem.


\begin{theorem}\label{Th_ortho_semi_ortho_rel}
    Let $M_1=(c_{ij})$ be a matrix of order $n$ over $\mathbb{F}_{2^m}^*$ of the form given in~(\ref{Eqn_The_matrixM1}) and $M_2=(d_{ij})$ be its inverse. Then $M=\Phi(D_1,D_2,M_1)$ is an orthogonal matrix if and only if $M_1$ is a semi-orthogonal matrix such that
	\begin{center}
		$M_1^{-T}=D_1^2M_1D_2^2$.
	\end{center}
	Moreover, $D_1$ and $D_2$ must take the following form:
	\begin{center}
		$D_1=\Diag(\sqrt{d_{11}},\sqrt{d_{12}},\ldots,\sqrt{d_{1n}})$ and $D_2=\Diag(1,\sqrt{\frac{d_{21}}{d_{11}}},\ldots,\sqrt{\frac{d_{n1}}{d_{11}}})$.
	\end{center}
\end{theorem}

\begin{proof}
        Let $D_1=\Diag(\lambda_1,\lambda_2,\ldots,\lambda_n)$ and $D_2=\Diag(1,\theta_2,\ldots,\theta_n)$ be two non-singular diagonal matrices over $\mathbb{F}_{2^m}$. Suppose that $M=\Phi(D_1,D_2,M_1)$ is an orthogonal matrix. Thus, we have
	\begin{eqnarray}
		&& MM^T=I \nonumber\\
		\implies && (D_1M_1D_2)(D_2M_1^TD_1)= I \nonumber\\
		\implies && D_1M_1D_2^2M_1^T= D_1^{-1} \nonumber\\
		\implies && D_1^2M_1D_2^2M_1^T= I \nonumber \\
		\implies && M_1^{-T}=D_1^2M_1D_2^2. \nonumber
	\end{eqnarray}

	Since $D_1^2$ and $D_2^2$ are non-singular diagonal matrices, we can conclude that $M_1$ is a semi-orthogonal matrix. Moreover, since $M_2 = (d_{ij})$ is the inverse of $M_1$, we also have the relation
	$$d_{ji} = \lambda_i^2 c_{ij} \theta_j^2.$$

	\noindent If we compare the first row and first column entries, since $c_{i1}=c_{1j}=1$ for $1\leq i,j \leq n$, we have
	
	\begin{equation*}
		\begin{aligned}
			d_{i1}&=  \lambda_1^2\theta_i^2 ~\text{ and }~ d_{1j}=  \theta_1^2\lambda_j^2 ~\text{ for }~ 1\leq i,j \leq n.
		\end{aligned}
	\end{equation*}

\noindent Since $\theta_1=1$, it follows that $d_{11}=\lambda_1^2$, which implies that $d_{11}\neq 0$. Furthermore, we have the following relations
\begin{equation*}
	\begin{aligned}
		\lambda_i=\sqrt{d_{1i}} ~\text{  and  }~ \theta_i=\sqrt{\frac{d_{i1}}{d_{11}}}~\text{  for  }~ 1\leq i\leq n.
	\end{aligned}
\end{equation*}

It is important to note that we are working in a finite field of characteristic $2$, which ensures the existence of square roots in the above equations. Consequently, the diagonal matrices $D_1$ and $D_2$ assume the desired form as stated in the theorem.\\

\noindent Conversely, suppose that $M_1$ is a semi-orthogonal matrix such that
	\begin{eqnarray*}
	&&M_1^{-T}=D_1^2M_1D_2^2,
	\end{eqnarray*}
where $D_1$ and $D_2$ has the following form
\begin{center}
$D_1=\Diag(\sqrt{d_{11}},\sqrt{d_{12}},\ldots,\sqrt{d_{1n}})$ and $D_2=\Diag(1,\sqrt{\frac{d_{21}}{d_{11}}},\ldots,\sqrt{\frac{d_{n1}}{d_{11}}})$.
\end{center}

Thus, we obtain the following expression
\begin{equation}\label{Eqn_M1_SOMDS}
\begin{aligned}
    d_{ji} &= \lambda_i^2 c_{ij} \theta_j^2 \quad \text{for} \quad 1 \leq i, j \leq n.
\end{aligned}
\end{equation}

Now, we will show that $M=\Phi(D_1,D_2,M_1)$ is an orthogonal matrix. If $M=(m_{ij})$ then we have $m_{ij}=\lambda_i \theta_j c_{ij}$ for $1\leq i,j \leq n$. Therefore, we have

\begin{equation*}
    \begin{aligned}
        (MM^T)_{ij} &= \sum_{k=1}^n m_{ik}m_{jk} \\
        &= \sum_{k=1}^n (\lambda_i \theta_k c_{ik})(\lambda_j \theta_k c_{jk}) \\
        &= \sum_{k=1}^n \lambda_i \lambda_j c_{ik} \theta_k^2 c_{jk} \\
        & =\sum_{k=1}^n \lambda_i \lambda_j c_{ik} \frac{d_{kj}}{\lambda_j^2} \quad [\text{From Equation}~\ref{Eqn_M1_SOMDS}] \\
        &=\frac{\lambda_i}{\lambda_j}\sum_{k=1}^n c_{ik}d_{kj} \\
        &=
        \begin{cases}
        0 & \text{if } i \neq j, \\
        1 & \text{if } i = j. \\
        \end{cases} \quad \quad [\text{Since} ~M_2=M_1^{-1}]
    \end{aligned}
\end{equation*}

\noindent This demonstrates that $MM^T=I$, implying that $M$ is an orthogonal matrix. This completes the proof. 
\qed
\end{proof}

\begin{corollary}\label{Coro_SO}
	The number of orthogonal matrices of order $n$ over $\mathbb{F}_{2^m}^{*}$ is equal to the number of semi-orthogonal matrices of order $n$ over $\mathbb{F}_{2^m}^{*}$ of the form $M_1$, as specified in (\ref{Eqn_The_matrixM1}).
\end{corollary}

\begin{proof}
	Let $M$ be an orthogonal matrix of order $n$ over $\mathbb{F}_{2^m}^{*}$. By Theorem~\ref{Th_decomposition}, there exists a unique matrix $M_1$ in the form of~(\ref{Eqn_The_matrixM1}) such that $M = D_1 M_1 D_2$, where $D_1$ and $D_2$ are non-singular diagonal matrices.

	According to Theorem~\ref{Th_ortho_semi_ortho_rel}, $M_1$ is a semi-orthogonal matrix, with $D_1^2$ and $D_2^2$ as the corresponding diagonal matrices. Furthermore, $D_1$ and $D_2$ assume the specific forms described in Theorem~\ref{Th_ortho_semi_ortho_rel}, which ensures that the choice of $D_1$ and $D_2$ is unique when $M$ is orthogonal.

	Thus, for a given semi-orthogonal matrix $M_1$, there is a unique construction of an orthogonal matrix $M$. Consequently, the total number of orthogonal matrices over $\mathbb{F}_{2^m}^{*}$ is equal to the number of semi-orthogonal matrices of order $n$ over $\mathbb{F}_{2^m}^{*}$ of the form $M_1$.
    \qed
\end{proof}

If $M = \Phi(D_1, D_2, M_1)$ is an MDS matrix of order $n$ over $\mathbb{F}_{2^m}$, then by Lemma~\ref{Lemma_DMD_MDS}, $M_1$ is also an MDS matrix. Furthermore, if $M$ is an orthogonal MDS matrix, Theorem~\ref{Th_ortho_semi_ortho_rel} implies that $M_1$ is a semi-orthogonal MDS matrix. Therefore, based on Corollary~\ref{Coro_SO}, we can conclude the following result.

\begin{theorem}\label{Th_count_Orthogonal_MDS}
    The number of orthogonal MDS matrices of order $n$ over $\mathbb{F}_{2^m}$ is equal to the number of semi-orthogonal MDS matrices of order $n$ over $\mathbb{F}_{2^m}$ of the form $M_1$, as specified in (\ref{Eqn_The_matrixM1}).
\end{theorem}


We now turn our attention to the results concerning semi-orthogonal MDS matrices over $\mathbb{F}_{2^m}$. To facilitate this discussion, we present the following lemma. The proof of this lemma closely resembles that of Lemma~\ref{Lemma_complete_SI}. Nevertheless, for the sake of completeness, we provide it here.

\begin{lemma}\label{Lemma_complete_SO}
	The number of semi-orthogonal matrices of order $n$ over $\mathbb{F}_{2^m}^{*}$ is given by 
	\[
		(2^m - 1)^{2n-1} \cdot N_3,
	\]
	where $N_3$ denotes the number of semi-orthogonal matrices of order $n$ over $\mathbb{F}_{2^m}^{*}$ of the form $M_1$, as specified in (\ref{Eqn_The_matrixM1}).
\end{lemma}

\begin{proof}
	Let $M=\Phi(D_1,D_2,M_1)$ be a semi-orthogonal matrix over $\mathbb{F}_{2^m}^{*}$ with corresponding diagonal matrices $D$ and $D'$ . Thus, we have 

	\begin{equation*}
		\begin{aligned}
		& M^{-T} =DMD'\\
		\implies & M^{T}(DMD')=I\\
		\implies & D_2M_1^{T}D_1 D(D_1M_1D_2) D'= I \\
		\implies & M_1^{T}(DD_1^2M_1D_2^2D')=I.
		\end{aligned}
	\end{equation*}

	Thus, the inverse of the matrix $M_1$ is given by 
	\begin{equation*}
		\begin{aligned}
			&M_1^{-T}=DD_1^2M_1D_2^2D'.
		\end{aligned}
	\end{equation*}
	
	Now, since $DD_1^2$ and $D_2^2D'$ are non-singular diagonal matrices, we can say that $M_1$ is also a semi-orthogonal matrix.

    \noindent Conversely, suppose that $M_1$ is a semi-orthogonal matrix. Then, by Lemma~\ref{lemma_D1MD2_SI_SO}, we can say that $M=\Phi(D_1,D_2,M_1)=D_1M_1D_2$ is also a semi-orthogonal matrix.

	Therefore, we have established that $M=\Phi(D_1,D_2,M_1)$ is a semi-orthogonal matrix if and only if $M_1$ is a semi-orthogonal matrix.

	Now, by selecting different values for $D_1$ and $D_2$, it is possible to construct other semi-orthogonal matrices from $M_1$. According to Theorem~\ref{Th_decomposition}, we have $(2^m - 1)^{2n-1}$ possible choices for $D_1$ and $D_2$. Therefore, we conclude that the total number of semi-orthogonal matrices over $\mathbb{F}_{2^m}^*$ is given by
	\begin{equation*}
		\begin{aligned}
			(2^m - 1)^{2n-1} \cdot N_3,
		\end{aligned}
	\end{equation*}
	where $N_3$ represents the number of semi-orthogonal matrices of order $n$ over $\mathbb{F}_{2^m}^{*}$ of the form $M_1$.
    \qed
\end{proof}

\noindent Similar to Theorem~\ref{Th_count_SIMDS}, using the above lemma, we can easily derive the result for semi-orthogonal MDS matrices over $\mathbb{F}_{2^m}$. Thus, we provide the result without proof.

\begin{theorem}\label{Th_count_Semi-OMDS}
	The number of semi-orthogonal MDS matrices of order $n$ over $\mathbb{F}_{2^m}$ is given by
	\[
		(2^m - 1)^{2n-1} \cdot N_4,
	\]
	where $N_4$ denotes the number of semi-orthogonal MDS matrices of order $n$ over $\mathbb{F}_{2^m}$ of the form $M_1$, as specified in (\ref{Eqn_The_matrixM1}).
\end{theorem}

\noindent In the existing literature, a precise count of the $3 \times 3$ orthogonal MDS matrices over $\mathbb{F}_{2^m}$ has not been established. In Theorem~\ref{count_ortho_mds}, we derive an explicit formula for the enumeration of $3 \times 3$ orthogonal MDS matrices over $\mathbb{F}_{2^m}$. Furthermore, by utilizing Theorems~\ref{Th_count_Orthogonal_MDS} and \ref{Th_count_Semi-OMDS}, we obtain an explicit formula for counting the $3 \times 3$ semi-orthogonal MDS matrices over $\mathbb{F}_{2^m}$. 



In Theorem~\ref{Th_gen_Form_Ortho}, we establish a matrix structure for the $n \times n$ orthogonal matrices. Utilizing this structure, we are able to determine the count of $4 \times 4$ orthogonal MDS matrices over $\mathbb{F}_{2^m}$ for $m = 3$ and $4$. Subsequently, by applying Theorems~\ref{Th_count_Orthogonal_MDS} and \ref{Th_count_Semi-OMDS}, we derive the count of $4 \times 4$ semi-orthogonal MDS matrices over $\mathbb{F}_{2^m}$ for the same values of $m$. The results are presented in Table~\ref{Table:4_semi_ortho_MDS}.

\begin{table}[htbp]
    \centering
	\caption{Count of $4\times 4$ orthogonal and semi-orthogonal MDS Matrices}
	\label{Table:4_semi_ortho_MDS}
	\vspace{2mm}
	\begin{tabular}{|c|c|c|c|}\hline
		Finite Field & Semi-orthogonal of the form $M_1$ & Orthogonal & Semi-orthogonal  \\ \hline 
		$\mathbb{F}_{2^3}$    & $ 720$      &  $720$      & $7^7\times 720$\\ 
		$\mathbb{F}_{2^4}$    &  $1147440$  &  $1147440$  & $15^7\times  1147440$ \\
		\hline
	\end{tabular}
\end{table}

\subsection{The Interconnection between Semi-Involutory and Semi-Orthogonal Matrices}

In this section, we will discuss the case when a matrix possesses both the properties of being semi-involutory and semi-orthogonal. It is straightforward to observe that if a matrix $M$ is symmetric, then $M$ is semi-involutory if and only if it is semi-orthogonal. However, the question arises whether this relationship holds for matrices that are not necessarily symmetric. The following theorem aims to generalize this condition for matrices that exhibit both semi-orthogonal and semi-involutory properties.

\begin{theorem}~\cite{cheon2021semi}\label{semi-invo-ortho_cond}
Let $M$ be a non-singular matrix of order $n$. Suppose $M = DM^{T}D'$ for some non-singular diagonal matrices $D$ and $D'$. Then $M$ is semi-involutory if and only if $M$ is a semi-orthogonal matrix.
\end{theorem}


In Lemma~\ref{Lemma_complete_SI}, we establish that the matrix $M = \Phi(D_1, D_2, M_1)$ is a semi-involutory matrix if and only if $M_1$ is semi-involutory. Similarly, in Lemma~\ref{Lemma_complete_SO}, we demonstrate that $M = \Phi(D_1, D_2, M_1)$ is a semi-orthogonal matrix if and only if $M_1$ is semi-orthogonal. Therefore, to determine whether the matrix $M$ is both semi-involutory and semi-orthogonal, we must examine the case in which $M_1$ possesses both properties. The following theorem addresses this result.

\begin{theorem}\label{rep_mat_semi-invo-ortho_cond}
	Let $M_1=(c_{ij})$ be a matrix of order $n$ of the form given in~(\ref{Eqn_The_matrixM1}) and suppose that $M_1 = DM_1^{T}D'$ for some non-singular diagonal matrices $D$ and $D'$. Then $M_1$ is a symmetric matrix.
\end{theorem}

\begin{proof}
	Suppose that $M_1 = DM_1^{T}D'$ for some non-singular diagonal matrices $D$ and $D'$, where $D=\Diag(\splitatcommas{\alpha_1,\alpha_2,\ldots,\alpha_n})$ and $D'=\Diag(\beta_1,\beta_2,\ldots,\beta_n)$. Thus, we have 
    \begin{equation}\label{rep_mat_semi-invo-ortho_cond:eq1}
        c_{ij}=\alpha_ic_{ji}\beta_j \quad \text{ for } 1\leq i, j\leq n.
    \end{equation}

	If we compare the first row, first column, and diagonal entries, since $c_{i1}=c_{1j}=1$ for $1\leq i,j \leq n$, we have
	\begin{eqnarray*}
		&&\alpha_1\beta_j=1  \text{ for } 1\leq j\leq n   \\
		&&\alpha_i\beta_1=1  \text{ for } 1\leq i\leq n    \\
		&&\alpha_i\beta_i=1  \text{ for } 1\leq i\leq n.
	\end{eqnarray*}

	After solving this system of equations, we have
	\begin{center}
		$\alpha_1=\alpha_2=\cdots =\alpha_n, \text{ and } \beta_1=\beta_2=\cdots =\beta_n$.
	\end{center}

	By putting these values of $\alpha_i$'s and $\beta_i$'s for $1\leq i\leq n$  in (\ref{rep_mat_semi-invo-ortho_cond:eq1}) we get

	\begin{equation*}
		\begin{aligned}
			&c_{ij}=c_{ji} \quad \text{ for } 1\leq i,j\leq n.
		\end{aligned}
	\end{equation*}

	Therefore, we have $M_1=M_1^T$. This completes the proof.
    \qed
\end{proof}

The next Lemma shows that the semi-involutory matrix of order $3$ over $\mathbb{F}_{2^m}^*$ of the form $M_1$ is symmetric.

\begin{lemma}\label{Th_symm_semi-invo_ortho_3x3}
    Let $M_1$ be a semi-involutory matrix of order $3$ over $\mathbb{F}_{2^m}^*$. Then $M_1$ is a symmetric matrix.
\end{lemma} 
\begin{proof}
Let
\begin{equation*}
    \begin{aligned}
       M_1&
       =\begin{pmatrix}
           1 &1 &1 \\
           1 &a &b \\
           1 &c &d
        \end{pmatrix}
    \end{aligned}
\end{equation*}
be a semi-involutory matrix and $D=\Diag(\alpha_1,\alpha_2,\alpha_3)$ and $D'=\Diag(\splitatcommas{\beta_1,\beta_2,\beta_3})$ be the corresponding diagonal matrices. Thus, we have
\begin{eqnarray*}
&& DM_1D'=M_1^{-1} \nonumber\\
&\implies& DM_1D'M_1 =I. 
\end{eqnarray*}

Since $(DM_1D'M_1)_{12}=(DM_1D'M_1)_{21}=0$, we have
\[
\beta_1=a\beta_2+c\beta_3 \quad \text{ and }
\beta_1=a\beta_2+b\beta_3.
\]

Adding these two equations, we can derive that $b=c$.

Therefore, $M_1$ is a symmetric matrix. This completes the proof.
\qed
\end{proof}

\begin{remark}
    From Lemma~\ref{Th_symm_semi-invo_ortho_3x3}, we know that any $3\times 3$ semi-involutory of the form $M_1$ is always symmetric. Additionally, by applying Theorems \ref{semi-invo-ortho_cond} and \ref{rep_mat_semi-invo-ortho_cond}, we can conclude that the number of $3 \times 3$ matrices, of the form $M_1$, that are both semi-involutory and semi-orthogonal is equal to the number of $3 \times 3$ semi-involutory matrices of the form $M_1$ over $\mathbb{F}_{2^m}^*$.
\end{remark}

\begin{theorem}\label{count_3_SISO-MDS}
    For $m\geq3$, the count of all $3 \times 3$ MDS matrices possessing both semi-involutory and semi-orthogonal properties over $\mathbb{F}_{2^m}$ is $(2^m-1)^5(2^m-2)(2^m-4)$.
\end{theorem}



\begin{remark}
    It is worth mentioning that Lemma \ref{Th_symm_semi-invo_ortho_3x3} is not true for higher orders. For example, consider the matrix
    \[
    M_1=(c_{ij})
    =\begin{pmatrix}
    1  & 1 & 1 & 1 \\
    1  & \alpha^5   &  \alpha       & \alpha^4\\
    1  & \alpha^4   & \alpha^{10}   & \alpha^2 \\
    1  & \alpha^8   &\alpha^5       & \alpha^{11}
    \end{pmatrix}
    \]
    over $\mathbb{F}_{2^4}^*$, where $\alpha$ is a primitive element and a root of the constructing polynomial $x^4+x+1$. It can be checked that the matrix $M_1$ is semi-involutory with the corresponding diagonal matrices $D=(1,\alpha^{12},\alpha,\alpha^{11})$ and $D'=(\alpha^{14},\alpha^{11},1,\alpha^{10})$. However, the matrix $M_1$ is not symmetric. This implies that over $\mathbb{F}_{2^m}^*$, the count of higher order (say $\geq 4$) matrices that are both semi-involutory and semi-orthogonal may not be the same as the number of semi-involutory matrices of that order.
\end{remark}

In Table~\ref{Table:count_4_SIMD&SOMDS}, we present the count of MDS matrices of order $4$ that possess both semi-involutory and semi-orthogonal properties.

\begin{table}[htbp]
    \centering
    \caption{Count of $4\times 4$ MDS matrices with both semi-involutory and semi-orthogonal properties}
    \label{Table:count_4_SIMD&SOMDS}
    \vspace{2mm}
    \begin{tabular}{|c|c|c|}\hline
        Finite Field & The matrix $M_1$ with both properties & Total Count\\ \hline
        $\mathbb{F}_{2^3}$ & $48$        &$7^7\times48$\\   
        $\mathbb{F}_{2^4}$ & $11088$     &$15^7\times11088$\\ 
        \hline
    \end{tabular}
\end{table}

\begin{remark}
    From Tables~\ref{Table:count_4_SIMDS} and \ref{Table:count_4_SIMD&SOMDS}, we can observe that the number of $4 \times 4$ MDS matrices that are both semi-involutory and semi-orthogonal is equal to the number of $4 \times 4$ semi-involutory MDS matrices over $\mathbb{F}_{2^3}$. This equivalence arises because there are no non-symmetric semi-involutory MDS matrices of order $4$ over $\mathbb{F}_{2^3}$ of the form $M_1$. In contrast, among the semi-involutory MDS matrices of the form $M_1$ of order $4$ over $\mathbb{F}_{2^4}$, we find $60768$ non-symmetric matrices and $11088$ symmetric matrices. This discrepancy leads to a different count for the matrices over the field $\mathbb{F}_{2^4}$.
\end{remark}

\section{On the Structure of Orthogonal MDS Matrices}~\label{Sec:orthogonal_count}
In this section, we begin by examining the structure of orthogonal matrices of order $n$ over $\mathbb{F}_{2^m}$. We then establish that the number of orthogonal MDS matrices of order $3$ over $\mathbb{F}_{2^m}$ is given by $(2^m - 2)(2^m - 3)(2^m - 4)$. Additionally, we provide the explicit count of orthogonal MDS matrices of order $4$ over $\mathbb{F}_{2^m}$ for $m = 3$ and $m = 4$. In the following theorem, we derive the general form of orthogonal matrices of order $n$ over $\mathbb{F}_{2^m}$.

\begin{theorem}\label{Th_gen_Form_Ortho}
    Let $M = (m_{ij})$ be a $n \times n$ matrix over $\mathbb{F}_{2^m}$. Then $M$ is orthogonal if and only if $M$ can be expressed in the following form:
    $$\begin{pmatrix}
    m_{11} & m_{12}  & \ldots  & m_{1n-1}    & b_1+1\\
    m_{21} & m_{22}  & \ldots  & m_{2n-1}    & b_2+1 \\
        \vdots &\vdots   &  \ddots &\vdots       & \vdots\\
    m_{n-11} & m_{n-12}  & \ldots  & m_{n-1n-1}  & b_{n-1}+1 \\
    c_1+1  & c_2+1   & \dots   &c_{n-1}+1        & \sum_{i=1}^{n-1}(b_i+1)+1
    \end{pmatrix},$$ 
    where $b_i=\sum_{j=1}^{n-1}m_{ij}$, $c_j=\sum_{i=1}^{n-1}m_{ij}$, $1\leq i,j\leq n-1$ and $m_{ij}$'s satisfy the following equations for $1 \leq k< l \leq n-1$
    \begin{align}
    \sum_{i =1}^{n-1}\sum_{j =1, j \neq i}^{n-1}m_{ki}m_{lj}+b_k+b_l =1 \label{eq:ind_row1}\\
    \sum_{i =1}^{n-1}\sum_{j =1, j \neq i}^{n-1}m_{ik}m_{jl}+c_k+c_l =1. \label{eq:ind_col2}
    \end{align}
\end{theorem}

\begin{proof}
Let
\begin{equation*}
\begin{aligned}
M &
= \begin{pmatrix}
m_{11} & m_{12} &\ldots&  m_{1n}\\
m_{21} & m_{22} &\ldots&  m_{2n}\\
m_{31} & m_{32} &\ldots&  m_{3n}\\
    \vdots &\vdots  &\ddots&  \vdots\\
m_{n1} & m_{n2} &\ldots&  m_{nn}
\end{pmatrix} 
\end{aligned}
\end{equation*}        
be a $n\times n$ orthogonal matrix over $\mathbb{F}_{2^m}$. Since $M$ is orthogonal, then $MM^T=M^TM=I$, where $I$ is a $n\times n$ identity matrix. Let $R_i$ and $C_i$ be the $i$th row and column vectors of $M$. We have 
\begin{equation} \label{Th_Orth_mat:eq}
R_i\cdot R_j = C_i\cdot C_j= \delta_{ij}=
\begin{cases} 
1 &  \text{if } i=j, \\
0 &  \text{if } i\neq j.
\end{cases}   
\end{equation}

For $i=j$, we get the following values of $m_{in}$ and $m_{nj}$, $1 \leq i,j \leq n$.

\begin{equation}\label{Eqn:single_i=j}
    \left.
    \begin{array}{ll}
    m_{1n}&= \sum_{j=1}^{n-1}m_{1j}+1= b_1+1 \\ [10pt]
    &\vdots \\ [10pt]
    m_{n-1n}&= \sum_{j=1}^{n-1}m_{n-1j}+1= b_{n-1}+1\\[10pt]
    m_{n1}&= \sum_{i=1}^{n-1}m_{i1}+1= c_1+1 \\ [10pt]
    &\vdots \\ [10pt]
    m_{nn-1}&= \sum_{i=1}^{n-1}m_{in-1}+1= c_{n-1}+1 \\ [10pt]
    m_{nn}&= \sum_{i=1}^{n-1} (b_i+1)+1 = \sum_{j=1}^{n-1}(c_j+1)+1
    \end{array}
    \right \}
\end{equation}\\

where $b_i=\sum_{j=1}^{n-1}m_{ij}$ and $c_j=\sum_{i=1}^{n-1}m_{ij}$, $1\leq i,j\leq n-1$.\\

From (\ref{Th_Orth_mat:eq}), let's consider the dot product of any two distinct rows (columns) of the first $(n-1)$ rows (columns) of $M$ which is zero. Thus, for $1 \leq k < l \leq n-1$, we have the following equations 
\begin{eqnarray}
   R_k\cdot R_l = 0   \label{Th_gen_Form_Orth:eq8} \\
   C_k\cdot C_l = 0.   \label{Th_gen_Form_Orth:eq9} 
\end{eqnarray} 
      
The equations $R_i\cdot R_n = 0$ and $C_i \cdot C_n = 0$  for $1 \leq i \leq n-1$ are redundant because they can be derived from the set of equations (\ref{Eqn:single_i=j}), (\ref{Th_gen_Form_Orth:eq8}) and (\ref{Eqn:single_i=j}), (\ref{Th_gen_Form_Orth:eq9}), respectively.\\

Now, by substituting the values of $m_{kn}$ and $m_{ln}$ in  (\ref{Th_gen_Form_Orth:eq8}), we get the following relation

\begin{align*}
\sum_{i =1}^{n-1}\sum_{j =1, j \neq i}^{n-1}m_{ki}m_{lj}+\sum_{i =1}^{n-1}(m_{ki}+ 
m_{li})&=\sum_{i =1}^{n-1}\sum_{j =1, j \neq i}^{n-1}m_{ki}m_{lj}+b_k+b_l =1. 
\end{align*}

Similarly, by substituting the values of $m_{nk}$ and $m_{nl}$ in  (\ref{Th_gen_Form_Orth:eq9}), we get the following relation
\begin{align*}
\sum_{i =1}^{n-1}\sum_{j =1, j \neq i}^{n-1}m_{ik}m_{jl}+\sum_{i =1}^{n-1}(m_{ik}+ 
  m_{il})&=\sum_{i =1}^{n-1}\sum_{j =1, j \neq i}^{n-1}m_{ik}m_{jl}+c_k+c_l =1.
\end{align*}
Thus, we obtain the desired form for $M$. This completes the proof. 
\qed
\end{proof}

\noindent There have been numerous studies on $2^n \times 2^n$ MDS matrices that focus on efficient implementation. However, other dimensions have not been explored as thoroughly. Since sub-matrices of an MDS matrix are also MDS, studying the structure of $3 \times 3$ MDS matrices could be useful in generating higher order MDS matrices. In addition, $3 \times 3$ MDS matrices that are involutory or orthogonal can be valuable in the design of lightweight block ciphers and hash functions with non-standard block or key sizes. For example, in~\cite{barreto2007curupira}, the CURUPIRA block cipher, designed for constrained platforms, was introduced. CURUPIRA operates on 96-bit plaintext blocks with key sizes of 96, 144, or 192 bits. The diffusion layer of CURUPIRA employs a $3 \times 3$ MDS matrix over $\mathbb{F}_{2^8}$. In Theorem~\ref{count_ortho_mds}, we provide an explicit formula for counting $3 \times 3$ orthogonal MDS matrices over $\mathbb{F}_{2^m}$. To derive this formula, we first identify the generic form of $3 \times 3$ orthogonal matrices in Corollary~\ref{rep_ortho}, which is established using Theorem~\ref{Th_gen_Form_Ortho}. Additionally, by analyzing Equations \ref{eq:ind_row1} and \ref{eq:ind_col2}, we find that they reduce to a single equation 
\[
m_{11}m_{22}+m_{12}m_{21}+m_{11}+m_{12}+m_{21}+m_{22}=1 
\]
for a $3 \times 3$ orthogonal matrix over $\mathbb{F}_{2^m}$. This leads to the following characterization of such matrices.

\begin{corollary}\label{rep_ortho}
    Let $M = (m_{ij})$ be a $3 \times 3$ matrix over $\mathbb{F}_{2^m}$. Then $M$ is orthogonal if and only if $M$ can be expressed in the following form 
    \begin{equation*}
        \begin{aligned}
            \begin{pmatrix}
            m_{11} & m_{12}  & m_{11}+m_{12}+1 \\
            m_{21} & m_{22} & m_{21}+m_{22}+1 \\
            m_{11}+m_{21}+1 & m_{12}+m_{22}+1 & m_{11}+m_{12}+m_{21}+m_{22}+1
            \end{pmatrix},
        \end{aligned}
    \end{equation*}
    where $m_{11}m_{22}+m_{12}m_{21}+m_{11}+m_{12}+m_{21}+m_{22} = 1$.
\end{corollary}

The explicit formula for counting orthogonal MDS matrices of order $3$ over $\mathbb{F}_{2^m}$ is obtained in the following theorem.

\begin{theorem}\label{count_ortho_mds}
   For $m\geq 3$, the number of all  orthogonal MDS matrices of order  $3$ over $\mathbb{F}_{2^m}$ is $(2^m-2)(2^m-3)(2^m-4)$.  
\end{theorem}
\begin{proof}
    Let $M$ be a $3\times 3$ matrix of the form as defined in Corollary \ref{rep_ortho}. From the definition of an MDS matrix, every square sub-matrix of an MDS matrix is non-singular. Since the inverse of an orthogonal matrix is the transpose of that matrix, if all the elements of an orthogonal matrix are non-zero, then all the minors of order $1$ and $2$ of that matrix are non-zero. Therefore, $M$ is MDS if and only if all elements of $M$ are non-zero, i.e., 
    \begin{eqnarray}\label{eqn_ortho_count}
    && m_{11}, m_{12}, m_{21}, m_{22} \in \mathbb{F}_{2^m}^*, m_{12} \neq  m_{11}+1, m_{21} \neq  m_{11}+1 \nonumber \\
    && m_{22}\neq m_{12}+1, m_{22}\neq m_{21}+1, m_{22} \neq m_{11}+m_{12}+m_{21}+1.
    \end{eqnarray}
    
    \noindent We will now analyze two cases to determine the count of all orthogonal MDS matrices of order $3$.  \\
    
   \noindent \textbf{Case I:} $m_{11}=1$. From (\ref{eqn_ortho_count}), we have the following conditions  on the parameters of $M$
    \begin{eqnarray*}
    && m_{12}, m_{21}, m_{22} \in \mathbb{F}_{2^m}^*, m_{22} \neq  m_{12}+1, m_{22} \neq  m_{21}+1\\
    && m_{22}\neq m_{12}+m_{21}.
    \end{eqnarray*}
    
    The parameters of $M$ also satisfy the following relation
    \begin{equation*}
    m_{11}m_{22}+m_{12}m_{21}+m_{11}+m_{12}+m_{21}+m_{22} = 1.    
    \end{equation*}

    Since $m_{11} = 1$, it follows that $m_{12}m_{21} = m_{12} + m_{21}$. Now, suppose $m_{12} = 1$. Substituting this into the equation, we get $m_{21} = 1 + m_{21}$, which leads to a contradiction. Therefore, $m_{12} \neq 1$. Consequently, $m_{21} = (m_{12} + 1)^{-1}m_{12}$.


    In this case, the number of all orthogonal MDS matrices of order $3$ over $\mathbb{F}_{2^m}$ is equal to the cardinality of the set $S_1$, defined as:
    
    \begin{center}
    \begin{small}
        $S_1= \big\{\splitatcommas{(m_{11}, m_{12}, m_{21}, m_{22})\in \mathbb{F}_{2^m}^4:~m_{11}=1, 
             m_{12} \not \in \{0,1\}, m_{21}=(m_{12}+1)^{-1}m_{12}, m_{22} \neq 0, 
             m_{22}\neq m_{12}+1, m_{22}\neq  (m_{12}+1)^{-1}m_{12}+1,
             m_{22}\neq (m_{12}+1)^{-1}m_{12}+m_{12}
             } \big \}$.
    \end{small}
    \end{center}
	
    Since $m_{11}=1$ and $m_{21}=(m_{12}+1)^{-1}m_{12}$, then the problem of finding the cardinality of $S_1$ is reduced to find the count of two parameters $m_{12}$ and $m_{22}$.
    
    Now, since $m_{12}\in \mathbb{F}_{2^m}\setminus \{0,1\}$, we have $2^m-2$ many choices for $m_{12}$.
    
    Also, since $m_{22} \neq 0$,  $m_{22}\neq m_{12}+1$, $m_{22}\neq (m_{12}+1)^{-1}m_{12}+1$, and $m_{22}\neq (m_{12}+1)^{-1}m_{12}+m_{12}$, we can say that $m_{22}\notin T_1$, where $T_1$ is defined as
    $$T_1= \set{0,m_{12}+1,(m_{12}+1)^{-1}m_{12}+1,(m_{12}+1)^{-1}m_{12}+m_{12}}.$$
    
    We now demonstrate that the cardinality of $T_1$ is $4$ for any $m_{12} \in \mathbb{F}_{2^m} \setminus \set{0,1}$. To establish this, we equate all pairs of elements in $T_1$ and solve for $m_{12}$. The resulting cases are as follows:
       
   \begin{enumerate}[(i)]
    \setlength{\itemsep}{-1em}
       \item $m_{12}+1=0 \implies m_{12}=1$,  \\
       \item $(m_{12}+1)^{-1}m_{12}+1=0 \implies  m_{12}=m_{12}+1$,\\
       \item $(m_{12}+1)^{-1}m_{12}+m_{12}=0 \implies  m_{12}=0$, \\
       \item $m_{12}+1=(m_{12}+1)^{-1}m_{12}+1 \implies  m_{12}=0$, \\
       \item $m_{12}+1=(m_{12}+1)^{-1}m_{12}+m_{12} \implies m_{12}=m_{12}+1$,\\
       \item $(m_{12}+1)^{-1}m_{12}+1=(m_{12}+1)^{-1}m_{12}+m_{12} \implies m_{12}=1$.
   \end{enumerate}
   
    By analyzing these cases, we confirm that $T_1$ contains precisely four distinct elements. Hence, we have $2^m-4$ many choices for $m_{22}$. Therefore, the cardinality of $S_1$ is $(2^m-2)(2^m-4)$.\\
      
   \noindent \textbf{Case II:} $m_{11}\neq 1$. From (\ref{eqn_ortho_count}) we have the following conditions on the parameters of $M$
    \begin{eqnarray*} 
    && m_{11}, m_{12}, m_{21}, m_{22} \in \mathbb{F}_{2^m}^*, m_{12} \neq  m_{11}+1, m_{21} \neq  m_{11}+1 \nonumber \\
    && m_{22}\neq m_{12}+1, m_{22}\neq m_{21}+1, m_{22} \neq m_{11}+m_{12}+m_{21}+1.
   \end{eqnarray*}

    The parameters of $M$ also satisfy the following relation
    \begin{eqnarray} 
        && m_{11}m_{22}+m_{12}m_{21}+m_{11}+m_{12}+m_{21}+m_{22}=1 \nonumber \\
        && \implies m_{22}= (m_{11}+1)^{-1}(m_{12}m_{21}+m_{11}+m_{12}+m_{21}+1).  
    \end{eqnarray}
	
    Now, we will eliminate some restrictions on $m_{22}$. 
    Since $m_{22} \neq 0$ and $m_{22}= (m_{11}+1)^{-1}(m_{12}m_{21}+m_{11}+m_{12}+m_{21}+1)$, we have
    \begin{eqnarray*}
        && m_{12}m_{21}+m_{11}+m_{12}+m_{21}+1 \neq 0\\
        &\implies& m_{21}(m_{12}+1)\neq m_{11}+m_{12}+1.
    \end{eqnarray*}
    
    Also, $m_{22} \neq m_{12}+1$ implies that
    \begin{eqnarray*}
        && m_{12}m_{21}+m_{11}+m_{12}+m_{21}+1 \neq (m_{11}+1) (m_{12}+1) \\
        &\implies& m_{12}m_{21}+m_{11}+m_{12}+m_{21}+1 \neq m_{11}m_{12}+m_{11}+m_{12}+1 \\
        &\implies& m_{21}(m_{12}+1) \neq m_{11}m_{12}. 
    \end{eqnarray*}
    
    Also, when $m_{22} \neq m_{21}+1$ implies that
    \begin{eqnarray*}
        && m_{12}m_{21}+m_{11}+m_{12}+m_{21}+1 \neq (m_{11}+1) (m_{21}+1) \\
        &\implies& m_{12}m_{21}+m_{11}+m_{12}+m_{21}+1 \neq m_{11}m_{21}+m_{11}+m_{21}+1 \\
        &\implies& m_{21}(m_{12}+m_{11}) \neq m_{12}.
    \end{eqnarray*}
    
    Similarly, $m_{22} \neq m_{11}+m_{12}+m_{21}+1$ implies that  
    \begin{eqnarray*}
        &\implies&  m_{12}m_{21}+m_{11}+m_{12}+m_{21}+1 \neq (m_{11}+1) (m_{11}+m_{12}+m_{21}+1) \\
        &\implies&  m_{12}m_{21} \neq m_{11}(m_{11}+m_{12}+m_{21}+1) \\ 
        &\implies&  m_{21}(m_{12}+m_{11}) \neq m_{11}(m_{11}+m_{12}+1).
    \end{eqnarray*}
    
    Now, in this case, the number of all orthogonal MDS matrices of order  $3$ over $\mathbb{F}_{2^m}$ is equal to the cardinality of the set $S_2$, defined as:

    \begin{small}
        \begin{equation}\label{Eqn:Set-S2}
            \begin{aligned}
                S_2 = \Big\{ 
                & (m_{11}, m_{12}, m_{21}, m_{22}) \in \mathbb{F}_{2^m}^4:~
                m_{11} \notin \{0,1\},~m_{12} \neq 0,~m_{12} \neq m_{11} + 1,~m_{21} \neq 0, \\
                & ~m_{21} \neq m_{11} + 1,~m_{21}(m_{12}+1)\neq m_{11}+m_{12}+1,~m_{21}(m_{12} + 1) \neq m_{11}m_{12},\\
                &~m_{21}(m_{12} + m_{11}) \neq m_{12},~m_{21}(m_{12} + m_{11}) \neq m_{11}(m_{11} + m_{12} + 1), \\
                &~m_{22} = (m_{11} + 1)^{-1}(m_{12}m_{21} + m_{11} + m_{12} + m_{21} + 1)
                \Big\}.
            \end{aligned}
        \end{equation}
    \end{small}

    
    Since $m_{22}= (m_{11}+1)^{-1}(m_{12}m_{21}+m_{11}+m_{12}+m_{21}+1)$, then the problem of finding the cardinality of $S_2$ is reduced to find the count of three parameters $m_{11}$, $m_{12}$ and $m_{21}$. For this, we will analyze two subcases further. \\

   \noindent \textbf{Subcase 1:} $m_{12}=1$. In this Subcase, the set $S_2$, given in (\ref{Eqn:Set-S2}), is reduced to
   \begin{center}
        \begin{small}
            $S_2=\big \{\splitatcommas{(m_{11}, m_{12}, m_{21},m_{22})\in \mathbb{F}_{2^m}^4:~m_{11} \notin 
                 \set{0,1},~m_{12}=1,~m_{21}\neq 0,~m_{21} \neq  m_{11}+1,
                 ~m_{21} \neq (m_{11}+1)^{-1},~m_{21} \neq (m_{11}+1)^{-1}m_{11}^2,
                 ~m_{22}= (m_{11}+1)^{-1}m_{11}
                 } \big \}$.
        \end{small}
    \end{center}
 
       Now, since $m_{11}\in \mathbb{F}_{2^m}\setminus \{0,1\}$, we have $2^m-2$ many choices for $m_{11}$.\\
    
       Also since $m_{21}\neq 0$,  $m_{21} \neq  m_{11}+1$, $m_{21} \neq (m_{11}+1)^{-1}$, and   $m_{21} \neq (m_{11}+1)^{-1}m_{11}^2$, we can say that $m_{21}\notin T_2$, where $T_2$ is defined as
       $$T_2= \set{0,m_{11}+1,(m_{11}+1)^{-1}, (m_{11}+1)^{-1}m_{11}^2}.$$

       We now demonstrate that the cardinality of $T_2$ is $4$ for any values of $m_{11} \in \mathbb{F}_{2^m}\setminus\{0,1\}$. To establish this, we equate all pairs of elements in $T_1$ and solve for $m_{11}$. By analyzing these six cases, we confirm that $T_2$ contains precisely four distinct elements. Thus, for $m_{21}$, we have $2^m-4$ many choices over $\mathbb{F}_{2^m}$.\\
       
       Hence, in this Subcase, the cardinality of $S_2$ is $(2^m-2)(2^m-4)$.\\
       
      \noindent \textbf{Subcase 2:} $m_{12}\neq 1$. In this subcase, we will analyze two subcases further. \\
    
       \textbf{Subcase 2.1:} $m_{11}=m_{12}$. In this Subcase, the set $S_2$, given in (\ref{Eqn:Set-S2}), is reduced to
       \begin{center}
		\begin{small}
			$S_2=\big \{\splitatcommas{(m_{11}, m_{12}, m_{21},m_{22})\in \mathbb{F}_{2^m}^4:~m_{11} \notin \set{0,1},
             ~m_{12}=m_{11}, ~m_{21}\neq 0, ~m_{21} \neq  m_{11}+1,
             ~m_{21} \neq (m_{11}+1)^{-1},  ~m_{21} \neq (m_{11}+1)^{-1}m_{11}^2,
             ~m_{22}= (m_{11}+1)^{-1}(m_{11}m_{21}+m_{21}+1)
             } \big \}$.
		\end{small}
	\end{center}
 
       Now, since $m_{11}\in \mathbb{F}_{2^m}\setminus \{0,1\}$, we have $2^m-2$ many choices for $m_{11}$.\\

    Also, we have $m_{21} \neq 0$, $m_{21} \neq m_{11} + 1$, $m_{21} \neq (m_{11} + 1)^{-1}$, and $m_{21} \neq (m_{11} + 1)^{-1}m_{11}^2$. Therefore, by the similar analysis done for Subcase 1, we can conclude that we have $2^m - 4$ many choices for $m_{21}$.\\
       
   Hence, in Subcase 2.1, the cardinality of $S_2$ is $(2^m-2)(2^m-4)$.\\

       \textbf{Subcase 2.2:} $m_{11} \neq m_{12}$. In this Subcase, the set $S_2$, given in (\ref{Eqn:Set-S2}), is reduced to
       \begin{center}
		\begin{small}
			$S_2=\big \{\splitatcommas{(m_{11}, m_{12}, m_{21},m_{22})\in \mathbb{F}_{2^m}^4:~m_{11},m_{12} 
                 \notin \set{0,1},
                 ~m_{12}\neq m_{11}, ~m_{12} \neq  m_{11}+1, ~m_{21}\neq 0, ~m_{21} \neq  m_{11}+1,
                 ~m_{21} \neq (m_{12}+1)^{-1}(m_{11}+m_{12}+1),  ~m_{21} \neq (m_{12}+1)^{-1}m_{11}m_{12},
                 ~m_{21} \neq (m_{11}+m_{12})^{-1}m_{12},
                 ~m_{21} \neq m_{11}(m_{11}+m_{12})^{-1}(m_{11}+m_{12}+1),
                 ~m_{22}= (m_{11}+1)^{-1}(m_{12}m_{21}+m_{11}+m_{12}+m_{21}+1)
                 } \big \}$.
		\end{small}
	\end{center}
 
       Now, since $m_{11}\in \mathbb{F}_{2^m}\setminus \{0,1\}$, we have $2^m-2$ many choices for $m_{11}$.\\

       Also since $m_{12}\in \mathbb{F}_{2^m}\setminus \{0,1\}$,  $m_{12} \neq  m_{11}$,  and $m_{12} \neq m_{11}+1$, for $m_{12}$, we have $2^m-4$ many choices.\\
    
       Also since $~m_{21}\neq 0$,  $m_{21} \neq  m_{11}+1$,  $m_{21} \neq (m_{12}+1)^{-1}(m_{11}+m_{12}+1)$,  $m_{21} \neq (m_{12}+1)^{-1}m_{11}m_{12}$, $m_{21} \neq (m_{11}+m_{12})^{-1}m_{12}$, and 
     $m_{21} \neq m_{11}(m_{11}+m_{12})^{-1}(m_{11}+m_{12}+1)$, we can say that $m_{21} \notin T_3$, where $T_3$ is defined as
        \begin{equation*}
            \begin{split}
             &T_3= \{0,m_{11}+1, (m_{12}+1)^{-1}(m_{11}+m_{12}+1), (m_{12}+1)^{-1}m_{11}m_{12},\\
             &\qquad (m_{11}+m_{12})^{-1}m_{12}, m_{11}(m_{11}+m_{12})^{-1}(m_{11}+m_{12}+1)
             \}.
            \end{split}
        \end{equation*}

    We now demonstrate that the cardinality of $T_3$ is $6$ for any values of $m_{11} \notin \mathbb{F}_{2^m}\setminus\{0,1\}$ and $m_{12} \notin \set{0,1,m_{11},m_{11}+1}$. To establish this, we equate all pairs of elements in $T_3$ and solve for $m_{11}$ or $m_{12}$. By analyzing these $15$ cases, we can say that $T_3$ contains precisely six elements. Therefore, we have $2^m-6$ many choices for $m_{21}$.\\

   Hence, in Subcase 2.2, the cardinality of $S_2$ is $(2^m-2)(2^m-4)(2^m-6)$.\\

   Therefore, in Subcase 2, the cardinality of $S_2$ is given by
   \begin{equation*}
       \begin{aligned}
           &(2^m-2)(2^m-4)+(2^m-2)(2^m-4)(2^m-6)\\
           &=(2^m-2)(2^m-4)(2^m-5).
       \end{aligned}
   \end{equation*}

The cardinality of $S_2$ in \textbf{Case II} is determined by whether $m_{12} = 1$ (Subcase 1) or $m_{12} \neq 1$ (Subcase 2), and is given by        
\begin{equation*}
   \begin{aligned}
       &(2^m-2)(2^m-4)+(2^m-2)(2^m-4)(2^m-5)\\
       &=(2^m-2)(2^m-4)^2.
   \end{aligned}
\end{equation*}

Finally, by analyzing \textbf{Case I}  and \textbf{Case II}, we can conclude that the total number of orthogonal MDS matrices of order $3$ over $\mathbb{F}_{2^m}$ is
\begin{equation*}
   \begin{aligned}
       &(2^m-2)(2^m-4)+(2^m-2)(2^m-4)^2\\
       &=(2^m-2)(2^m-3)(2^m-4).
   \end{aligned}
\end{equation*}	
This concludes the proof.
\qed
\end{proof}

The number of all $3 \times 3$ orthogonal MDS matrices is $(2^m-2)(2^m-3)(2^m-4)$, as demonstrated in Theorem \ref{count_ortho_mds}. Hence, from Theorem \ref{Th_count_Orthogonal_MDS}, the number of all $3\times 3$ semi-orthogonal representative MDS matrices of the form $M_1$ is $(2^m-2)(2^m-3)(2^m-4)$. By the definition of a semi-orthogonal matrix, for any two non-singular diagonal matrices $D_1$ and $D_2$, if $M_1$ is semi-orthogonal, then $M=\Phi(D_1,D_2,M_1)$ is also a semi-orthogonal matrix. Hence, choices for $D_1$ and $D_2$ for $M$ to be a semi-orthogonal matrix are $(2^m-1)^5$, therefore, the number of all $3\times 3$ semi-orthogonal MDS  matrices is $(2^m-1)^5(2^m-2)(2^m-3)(2^m-4)$. We formally present this result below.

\begin{theorem}\label{Th_count_SOMDS_3}
    For $m\geq 3$, the number of semi-orthogonal MDS matrices of order $3$ over $\mathbb{F}_{2^m}$ is $(2^m-1)^5(2^m-2)(2^m-3)(2^m-4)$.
\end{theorem}

\begin{remark}
    We note that in a recent work~\cite{kumar2026study}, the authors derived a formula for counting $3 \times 3$ semi-orthogonal MDS matrices over $\mathbb{F}_{2^m}$. We derive the general structure of $n\times n$ orthogonal matrices over $\mathbb{F}_{2^m}$, which yields the formula $(2^m - 2)(2^m - 3)(2^m - 4)$ for $3 \times 3$ orthogonal MDS matrices. The structural interconnections we establish in Theorem~\ref{Th_count_Orthogonal_MDS} between semi-orthogonal and orthogonal matrices provide a theoretical link between our enumeration and their result.
\end{remark}


\noindent From Theorem~\ref{Th_gen_Form_Ortho}, we know that any $4\times 4$ orthogonal matrix over $\mathbb{F}_{2^m}$ can be expressed in the following form:
\begin{equation}\label{Eqn:GenForm-4Ortho}
    \begin{pmatrix}
    m_{11} & m_{12}  & m_{13}  & b_1+1\\
    m_{21} & m_{22} & m_{23} & b_2+1 \\
    m_{31} & m_{32} & m_{33} & b_3+1 \\
    c_1+1 & c_2+1& c_3+1 & b_1+b_2+b_3
\end{pmatrix},
\end{equation}
where $b_i=\sum_{j=1}^{3}m_{ij}$, $c_j=\sum_{i=1}^{3}m_{ij}$, $1\leq i,j \leq 3$ and $m_{ij}$'s satisfy the following equations for $1 \leq k< l \leq 3$
\begin{align*}
    \sum_{i =1}^{3}\sum_{j =1, j \neq i}^{3}m_{ki}m_{lj}+b_k+b_l =1 \\
    \sum_{i =1}^{3}\sum_{j =1, j \neq i}^{3}m_{ik}m_{jl}+c_k+c_l =1.
\end{align*}

Now, we are ready to get a count of the $4\times 4$ orthogonal MDS matrices over $\mathbb{F}_{2^m}$. The result is given in Table~\ref{Table:4_semi_ortho_MDS}. Note that because of the large search space involved, we are unable to offer the count of all $4\times 4$ orthogonal MDS matrices over $\mathbb{F}_{2^m}$ for $m\geq 5$. Nevertheless, no explicit count is currently available for all $4\times 4$ orthogonal MDS matrices. Hence, an explicit formula for counting the orthogonal MDS matrices of order $n\geq 4$ could be a possible direction for future research.

\begin{remark}
   It is worth noting that there is an additional advantage to constructing an orthogonal MDS matrix. We know that if all of the entries in the inverse of a $4 \times 4$ non-singular matrix are non-zero, then all of its $3 \times 3$ sub-matrices are also non-singular (in this case, $M^{-1}=M^T$). Therefore, to confirm the MDS property of the orthogonal matrix of order $4$, we need to inspect sub-matrices of $M$ of order $1$ and $2$.
\end{remark}

\noindent Since $4 \times 4$ matrices are most commonly used in the design of block ciphers and hash functions, we focus on finding efficient orthogonal MDS matrices of order 4. To do so, we use the general form given in (\ref{Eqn:GenForm-4Ortho}) to identify the lightest orthogonal MDS matrices. In the past, it was widely believed that implementing the multiplication of finite field elements with low Hamming weights incurred lower hardware costs. However, in 2014, the authors of~\cite{dXOR_FOAM2014} proposed an approach for estimating implementation costs by counting the number of XOR gates (d-XOR gates) required to implement the field element based on the multiplicative matrix of the element. \\

We adopt this metric to identify lightweight orthogonal MDS matrices of order 4 over $\mathbb{F}_{2^3}$ and $\mathbb{F}_{2^4}$. We found that orthogonal MDS matrices of order 4 achieve the lowest d-XOR count of 64 over $\mathbb{F}_{2^3}$, with 144 such matrices exhibiting this optimal count. These matrices are provided in Appendix~\ref{Appendix:4-OMDS_F_2^3}. Similarly, we observed that orthogonal MDS matrices of order 4 achieve the lowest d-XOR count of 72 over $\mathbb{F}_{2^4}$, with 144 matrices that have the lowest d-XOR count. These matrices are listed in Appendix~\ref{Appendix:4-OMDS_F_2^4}.

\begin{table}[ht]
    \centering
    \caption{Comparison of the count of $3\times 3$ MDS matrices with various properties}
	\label{Table:Count_Comparison_3_MDS}
	\vspace{2mm}
	\begin{tabular}{|c|c|c|c|c|c|c|}\hline
           & IMDS   & SIMDS  & OMDS  & SOMDS & SISOMDS \\ \hline
           $\mathbb{F}_{2^3}$  & $7^2\times24$      & $7^5\times24$        & $120$      & $7^5\times 120$  & $7^5\times24$  \\   
           $\mathbb{F}_{2^4}$  & $15^2\times168$    & $15^5\times168$      & $2184$     & $15^5\times 2184$  & $15^5\times168$ \\                       
           $\mathbb{F}_{2^5}$  & $31^2\times840$    & $31^5\times840$      & $24360$    & $31^5\times 24360$  & $31^5\times840$ \\                      
           $\mathbb{F}_{2^6}$  & $63^2\times3720$   & $63^5\times3720$     & $226920$   & $63^5\times 226920$  & $63^5\times3720$ \\                   
           $\mathbb{F}_{2^7}$  & $127^2\times15624$ & $127^5\times15624$   & $1953000$  & $127^5\times 1953000$  & $127^5\times15624$ \\                  
           $\mathbb{F}_{2^8}$  & $255^2\times64008$ & $255^5\times64008$   & $16194024$ & $255^5\times 16194024$  & $255^5\times64008$ \\            
		\hline
	\end{tabular}
\end{table}


\begin{table}[ht]
    \centering
    \begin{threeparttable}
    \caption{Comparison of the count of $4\times 4$ MDS matrices with various properties~\tnote{1}}
    \label{Table:Count_Comparison_4_MDS}
    \vspace{2mm}
    \begin{tabular}{|c|c|c|c|c|c|}\hline
        &     IMDS   & SIMDS  & OMDS  & SOMDS & SISOMDS\\ \hline
        $\mathbb{F}_{2^3}$  &$7^3\times48$             &$7^7\times48$        & $720$     & $7^7\times 720$      &$7^7\times48$\\ 
        $\mathbb{F}_{2^4}$  &$15^3\times71856$         &$15^7\times71856$    & $1147440$ & $15^7\times 1147440$ &$15^7\times11088$ \\ 
        $\mathbb{F}_{2^5}$  &$31^3\times10188240$      &$31^7\times10188240$ &&&\\ 
        $\mathbb{F}_{2^6}$  &$63^3\times612203760$     &$63^7\times612203760$ &&&\\ 
        $\mathbb{F}_{2^7}$  &$127^3\times26149708368$  &$127^7\times26149708368$ &&&\\ 
        $\mathbb{F}_{2^8}$  &$255^3\times961006331376$ &$255^7\times961006331376$ &&&\\
        \hline
    \end{tabular}
    \begin{tablenotes}
            \item[1] The abbreviations are as follows. IMDS: Involutory MDS, SIMDS: Semi-involutory MDS, OMDS: Orthogonal MDS, SOMDS: Semi-orthogonal MDS, and SISOMDS: Semi-involutory and semi-orthogonal MDS. The same abbreviations apply to Table~\ref{Table:Count_Comparison_3_MDS}.
        \end{tablenotes}
    \end{threeparttable}
\end{table}


\section{Conclusion}\label{Sec:Conclusion}

This paper presents a study of the structural relationships between semi-involutory and involutory matrices, as well as between semi-orthogonal and orthogonal matrices. We show that these relationships allow us to derive the enumeration of semi-involutory MDS matrices from the count of involutory ones, and vice versa. A similar correspondence is demonstrated for the semi-orthogonal and orthogonal cases. Leveraging these results, we provide a simplified derivation of the formula for counting $3 \times 3$ semi-involutory MDS matrices. Additionally, we establish explicit formulas for enumerating $3 \times 3$ orthogonal and semi-orthogonal MDS matrices over $\mathbb{F}_{2^m}$. We also characterize the class of matrices that simultaneously satisfy both the semi-involutory and semi-orthogonal properties. Using this characterization, we prove that the number of $3 \times 3$ MDS matrices possessing both properties is identical to the count of $3 \times 3$ semi-involutory MDS matrices over $\mathbb{F}_{2^m}$. Finally, we extend our analysis to $4 \times 4$ matrices by determining the number of semi-involutory MDS matrices over $\mathbb{F}_{2^m}$ for $m = 3, 4, \ldots, 8$, as well as the counts of orthogonal and semi-orthogonal MDS matrices over $\mathbb{F}_{2^m}$ for $m = 3, 4$.

\medskip

\bibliographystyle{plain}
\bibliography{ref}

\appendix
\section{The $4\times 4$ orthogonal MDS matrices with optimal d-XOR count $64$ over $\mathbb{F}_{2^3}/0$x$b$} \label{Appendix:4-OMDS_F_2^3}

The tuples in the following are represented as $(m_{11}, m_{12}, m_{13}, m_{21}, m_{22}, m_{23}, m_{31}, m_{32}, m_{33})$, where $m_{ij}$ ($1 \leq i,j \leq 3$) are defined by equation (\ref{Eqn:GenForm-4Ortho}).

\begin{small}
\begin{longtable}{c c c}
$(1_x, 2_x, 4_x, 2_x, 1_x, 6_x, 4_x, 6_x, 1_x)$,& $(1_x, 2_x, 4_x, 2_x, 1_x, 6_x, 6_x, 4_x, 2_x)$,& $(1_x, 2_x, 4_x, 4_x, 6_x, 1_x, 2_x, 1_x, 6_x)$ \\ 
$(1_x, 2_x, 4_x, 4_x, 6_x, 1_x, 6_x, 4_x, 2_x)$,& $(1_x, 2_x, 4_x, 6_x, 4_x, 2_x, 2_x, 1_x, 6_x)$,& $(1_x, 2_x, 4_x, 6_x, 4_x, 2_x, 4_x, 6_x, 1_x)$ \\ 
$(1_x, 2_x, 6_x, 2_x, 1_x, 4_x, 4_x, 6_x, 2_x)$,& $(1_x, 2_x, 6_x, 2_x, 1_x, 4_x, 6_x, 4_x, 1_x)$,& $(1_x, 2_x, 6_x, 4_x, 6_x, 2_x, 2_x, 1_x, 4_x)$ \\ 
$(1_x, 2_x, 6_x, 4_x, 6_x, 2_x, 6_x, 4_x, 1_x)$,& $(1_x, 2_x, 6_x, 6_x, 4_x, 1_x, 2_x, 1_x, 4_x)$,& $(1_x, 2_x, 6_x, 6_x, 4_x, 1_x, 4_x, 6_x, 2_x)$ \\ 
$(1_x, 4_x, 2_x, 2_x, 6_x, 1_x, 4_x, 1_x, 6_x)$,& $(1_x, 4_x, 2_x, 2_x, 6_x, 1_x, 6_x, 2_x, 4_x)$,& $(1_x, 4_x, 2_x, 4_x, 1_x, 6_x, 2_x, 6_x, 1_x)$ \\ 
$(1_x, 4_x, 2_x, 4_x, 1_x, 6_x, 6_x, 2_x, 4_x)$,& $(1_x, 4_x, 2_x, 6_x, 2_x, 4_x, 2_x, 6_x, 1_x)$,& $(1_x, 4_x, 2_x, 6_x, 2_x, 4_x, 4_x, 1_x, 6_x)$ \\ 
$(1_x, 4_x, 6_x, 2_x, 6_x, 4_x, 4_x, 1_x, 2_x)$,& $(1_x, 4_x, 6_x, 2_x, 6_x, 4_x, 6_x, 2_x, 1_x)$,& $(1_x, 4_x, 6_x, 4_x, 1_x, 2_x, 2_x, 6_x, 4_x)$ \\ 
$(1_x, 4_x, 6_x, 4_x, 1_x, 2_x, 6_x, 2_x, 1_x)$,& $(1_x, 4_x, 6_x, 6_x, 2_x, 1_x, 2_x, 6_x, 4_x)$,& $(1_x, 4_x, 6_x, 6_x, 2_x, 1_x, 4_x, 1_x, 2_x)$ \\ 
$(1_x, 6_x, 2_x, 2_x, 4_x, 1_x, 4_x, 2_x, 6_x)$,& $(1_x, 6_x, 2_x, 2_x, 4_x, 1_x, 6_x, 1_x, 4_x)$,& $(1_x, 6_x, 2_x, 4_x, 2_x, 6_x, 2_x, 4_x, 1_x)$ \\ 
$(1_x, 6_x, 2_x, 4_x, 2_x, 6_x, 6_x, 1_x, 4_x)$,& $(1_x, 6_x, 2_x, 6_x, 1_x, 4_x, 2_x, 4_x, 1_x)$,& $(1_x, 6_x, 2_x, 6_x, 1_x, 4_x, 4_x, 2_x, 6_x)$ \\ 
$(1_x, 6_x, 4_x, 2_x, 4_x, 6_x, 4_x, 2_x, 1_x)$,& $(1_x, 6_x, 4_x, 2_x, 4_x, 6_x, 6_x, 1_x, 2_x)$,& $(1_x, 6_x, 4_x, 4_x, 2_x, 1_x, 2_x, 4_x, 6_x)$ \\ 
$(1_x, 6_x, 4_x, 4_x, 2_x, 1_x, 6_x, 1_x, 2_x)$,& $(1_x, 6_x, 4_x, 6_x, 1_x, 2_x, 2_x, 4_x, 6_x)$,& $(1_x, 6_x, 4_x, 6_x, 1_x, 2_x, 4_x, 2_x, 1_x)$ \\ 
$(2_x, 1_x, 4_x, 1_x, 2_x, 6_x, 4_x, 6_x, 2_x)$,& $(2_x, 1_x, 4_x, 1_x, 2_x, 6_x, 6_x, 4_x, 1_x)$,& $(2_x, 1_x, 4_x, 4_x, 6_x, 2_x, 1_x, 2_x, 6_x)$ \\ 
$(2_x, 1_x, 4_x, 4_x, 6_x, 2_x, 6_x, 4_x, 1_x)$,& $(2_x, 1_x, 4_x, 6_x, 4_x, 1_x, 1_x, 2_x, 6_x)$,& $(2_x, 1_x, 4_x, 6_x, 4_x, 1_x, 4_x, 6_x, 2_x)$ \\ 
$(2_x, 1_x, 6_x, 1_x, 2_x, 4_x, 4_x, 6_x, 1_x)$,& $(2_x, 1_x, 6_x, 1_x, 2_x, 4_x, 6_x, 4_x, 2_x)$,& $(2_x, 1_x, 6_x, 4_x, 6_x, 1_x, 1_x, 2_x, 4_x)$ \\ 
$(2_x, 1_x, 6_x, 4_x, 6_x, 1_x, 6_x, 4_x, 2_x)$,& $(2_x, 1_x, 6_x, 6_x, 4_x, 2_x, 1_x, 2_x, 4_x)$,& $(2_x, 1_x, 6_x, 6_x, 4_x, 2_x, 4_x, 6_x, 1_x)$ \\ 
$(2_x, 4_x, 1_x, 1_x, 6_x, 2_x, 4_x, 2_x, 6_x)$,& $(2_x, 4_x, 1_x, 1_x, 6_x, 2_x, 6_x, 1_x, 4_x)$,& $(2_x, 4_x, 1_x, 4_x, 2_x, 6_x, 1_x, 6_x, 2_x)$ \\ 
$(2_x, 4_x, 1_x, 4_x, 2_x, 6_x, 6_x, 1_x, 4_x)$,& $(2_x, 4_x, 1_x, 6_x, 1_x, 4_x, 1_x, 6_x, 2_x)$,& $(2_x, 4_x, 1_x, 6_x, 1_x, 4_x, 4_x, 2_x, 6_x)$ \\ 
$(2_x, 4_x, 6_x, 1_x, 6_x, 4_x, 4_x, 2_x, 1_x)$,& $(2_x, 4_x, 6_x, 1_x, 6_x, 4_x, 6_x, 1_x, 2_x)$,& $(2_x, 4_x, 6_x, 4_x, 2_x, 1_x, 1_x, 6_x, 4_x)$ \\ 
$(2_x, 4_x, 6_x, 4_x, 2_x, 1_x, 6_x, 1_x, 2_x)$,& $(2_x, 4_x, 6_x, 6_x, 1_x, 2_x, 1_x, 6_x, 4_x)$,& $(2_x, 4_x, 6_x, 6_x, 1_x, 2_x, 4_x, 2_x, 1_x)$ \\ 
$(2_x, 6_x, 1_x, 1_x, 4_x, 2_x, 4_x, 1_x, 6_x)$,& $(2_x, 6_x, 1_x, 1_x, 4_x, 2_x, 6_x, 2_x, 4_x)$,& $(2_x, 6_x, 1_x, 4_x, 1_x, 6_x, 1_x, 4_x, 2_x)$ \\ 
$(2_x, 6_x, 1_x, 4_x, 1_x, 6_x, 6_x, 2_x, 4_x)$,& $(2_x, 6_x, 1_x, 6_x, 2_x, 4_x, 1_x, 4_x, 2_x)$,& $(2_x, 6_x, 1_x, 6_x, 2_x, 4_x, 4_x, 1_x, 6_x)$ \\
$(2_x, 6_x, 4_x, 1_x, 4_x, 6_x, 4_x, 1_x, 2_x)$,& $(2_x, 6_x, 4_x, 1_x, 4_x, 6_x, 6_x, 2_x, 1_x)$,& $(2_x, 6_x, 4_x, 4_x, 1_x, 2_x, 1_x, 4_x, 6_x)$ \\ 
$(2_x, 6_x, 4_x, 4_x, 1_x, 2_x, 6_x, 2_x, 1_x)$,& $(2_x, 6_x, 4_x, 6_x, 2_x, 1_x, 1_x, 4_x, 6_x)$,& $(2_x, 6_x, 4_x, 6_x, 2_x, 1_x, 4_x, 1_x, 2_x)$ \\ 
$(4_x, 1_x, 2_x, 1_x, 4_x, 6_x, 2_x, 6_x, 4_x)$,& $(4_x, 1_x, 2_x, 1_x, 4_x, 6_x, 6_x, 2_x, 1_x)$,& $(4_x, 1_x, 2_x, 2_x, 6_x, 4_x, 1_x, 4_x, 6_x)$ \\ 
$(4_x, 1_x, 2_x, 2_x, 6_x, 4_x, 6_x, 2_x, 1_x)$,& $(4_x, 1_x, 2_x, 6_x, 2_x, 1_x, 1_x, 4_x, 6_x)$,& $(4_x, 1_x, 2_x, 6_x, 2_x, 1_x, 2_x, 6_x, 4_x)$ \\ 
$(4_x, 1_x, 6_x, 1_x, 4_x, 2_x, 2_x, 6_x, 1_x)$,& $(4_x, 1_x, 6_x, 1_x, 4_x, 2_x, 6_x, 2_x, 4_x)$,& $(4_x, 1_x, 6_x, 2_x, 6_x, 1_x, 1_x, 4_x, 2_x)$ \\ 
$(4_x, 1_x, 6_x, 2_x, 6_x, 1_x, 6_x, 2_x, 4_x)$,& $(4_x, 1_x, 6_x, 6_x, 2_x, 4_x, 1_x, 4_x, 2_x)$,& $(4_x, 1_x, 6_x, 6_x, 2_x, 4_x, 2_x, 6_x, 1_x)$ \\ 
$(4_x, 2_x, 1_x, 1_x, 6_x, 4_x, 2_x, 4_x, 6_x)$,& $(4_x, 2_x, 1_x, 1_x, 6_x, 4_x, 6_x, 1_x, 2_x)$,& $(4_x, 2_x, 1_x, 2_x, 4_x, 6_x, 1_x, 6_x, 4_x)$ \\ 
$(4_x, 2_x, 1_x, 2_x, 4_x, 6_x, 6_x, 1_x, 2_x)$,& $(4_x, 2_x, 1_x, 6_x, 1_x, 2_x, 1_x, 6_x, 4_x)$,& $(4_x, 2_x, 1_x, 6_x, 1_x, 2_x, 2_x, 4_x, 6_x)$ \\
$(4_x, 2_x, 6_x, 1_x, 6_x, 2_x, 2_x, 4_x, 1_x)$,& $(4_x, 2_x, 6_x, 1_x, 6_x, 2_x, 6_x, 1_x, 4_x)$,& $(4_x, 2_x, 6_x, 2_x, 4_x, 1_x, 1_x, 6_x, 2_x)$ \\ 
$(4_x, 2_x, 6_x, 2_x, 4_x, 1_x, 6_x, 1_x, 4_x)$,& $(4_x, 2_x, 6_x, 6_x, 1_x, 4_x, 1_x, 6_x, 2_x)$,& $(4_x, 2_x, 6_x, 6_x, 1_x, 4_x, 2_x, 4_x, 1_x)$ \\ 
$(4_x, 6_x, 1_x, 1_x, 2_x, 4_x, 2_x, 1_x, 6_x)$,& $(4_x, 6_x, 1_x, 1_x, 2_x, 4_x, 6_x, 4_x, 2_x)$,& $(4_x, 6_x, 1_x, 2_x, 1_x, 6_x, 1_x, 2_x, 4_x)$ \\ 
$(4_x, 6_x, 1_x, 2_x, 1_x, 6_x, 6_x, 4_x, 2_x)$,& $(4_x, 6_x, 1_x, 6_x, 4_x, 2_x, 1_x, 2_x, 4_x)$,& $(4_x, 6_x, 1_x, 6_x, 4_x, 2_x, 2_x, 1_x, 6_x)$ \\ 
$(4_x, 6_x, 2_x, 1_x, 2_x, 6_x, 2_x, 1_x, 4_x)$,& $(4_x, 6_x, 2_x, 1_x, 2_x, 6_x, 6_x, 4_x, 1_x)$,& $(4_x, 6_x, 2_x, 2_x, 1_x, 4_x, 1_x, 2_x, 6_x)$ \\ 
$(4_x, 6_x, 2_x, 2_x, 1_x, 4_x, 6_x, 4_x, 1_x)$,& $(4_x, 6_x, 2_x, 6_x, 4_x, 1_x, 1_x, 2_x, 6_x)$,& $(4_x, 6_x, 2_x, 6_x, 4_x, 1_x, 2_x, 1_x, 4_x)$ \\ 
$(6_x, 1_x, 2_x, 1_x, 6_x, 4_x, 2_x, 4_x, 6_x)$,& $(6_x, 1_x, 2_x, 1_x, 6_x, 4_x, 4_x, 2_x, 1_x)$,& $(6_x, 1_x, 2_x, 2_x, 4_x, 6_x, 1_x, 6_x, 4_x)$ \\ 
$(6_x, 1_x, 2_x, 2_x, 4_x, 6_x, 4_x, 2_x, 1_x)$,& $(6_x, 1_x, 2_x, 4_x, 2_x, 1_x, 1_x, 6_x, 4_x)$,& $(6_x, 1_x, 2_x, 4_x, 2_x, 1_x, 2_x, 4_x, 6_x)$ \\ 
$(6_x, 1_x, 4_x, 1_x, 6_x, 2_x, 2_x, 4_x, 1_x)$,& $(6_x, 1_x, 4_x, 1_x, 6_x, 2_x, 4_x, 2_x, 6_x)$,& $(6_x, 1_x, 4_x, 2_x, 4_x, 1_x, 1_x, 6_x, 2_x)$ \\ 
$(6_x, 1_x, 4_x, 2_x, 4_x, 1_x, 4_x, 2_x, 6_x)$,& $(6_x, 1_x, 4_x, 4_x, 2_x, 6_x, 1_x, 6_x, 2_x)$,& $(6_x, 1_x, 4_x, 4_x, 2_x, 6_x, 2_x, 4_x, 1_x)$ \\ 
$(6_x, 2_x, 1_x, 1_x, 4_x, 6_x, 2_x, 6_x, 4_x)$,& $(6_x, 2_x, 1_x, 1_x, 4_x, 6_x, 4_x, 1_x, 2_x)$,& $(6_x, 2_x, 1_x, 2_x, 6_x, 4_x, 1_x, 4_x, 6_x)$ \\ 
$(6_x, 2_x, 1_x, 2_x, 6_x, 4_x, 4_x, 1_x, 2_x)$,& $(6_x, 2_x, 1_x, 4_x, 1_x, 2_x, 1_x, 4_x, 6_x)$,& $(6_x, 2_x, 1_x, 4_x, 1_x, 2_x, 2_x, 6_x, 4_x)$ \\ 
$(6_x, 2_x, 4_x, 1_x, 4_x, 2_x, 2_x, 6_x, 1_x)$,& $(6_x, 2_x, 4_x, 1_x, 4_x, 2_x, 4_x, 1_x, 6_x)$,& $(6_x, 2_x, 4_x, 2_x, 6_x, 1_x, 1_x, 4_x, 2_x)$ \\ 
$(6_x, 2_x, 4_x, 2_x, 6_x, 1_x, 4_x, 1_x, 6_x)$,& $(6_x, 2_x, 4_x, 4_x, 1_x, 6_x, 1_x, 4_x, 2_x)$,& $(6_x, 2_x, 4_x, 4_x, 1_x, 6_x, 2_x, 6_x, 1_x)$ \\ 
$(6_x, 4_x, 1_x, 1_x, 2_x, 6_x, 2_x, 1_x, 4_x)$,& $(6_x, 4_x, 1_x, 1_x, 2_x, 6_x, 4_x, 6_x, 2_x)$,& $(6_x, 4_x, 1_x, 2_x, 1_x, 4_x, 1_x, 2_x, 6_x)$ \\ 
$(6_x, 4_x, 1_x, 2_x, 1_x, 4_x, 4_x, 6_x, 2_x)$,& $(6_x, 4_x, 1_x, 4_x, 6_x, 2_x, 1_x, 2_x, 6_x)$,& $(6_x, 4_x, 1_x, 4_x, 6_x, 2_x, 2_x, 1_x, 4_x)$ \\ 
$(6_x, 4_x, 2_x, 1_x, 2_x, 4_x, 2_x, 1_x, 6_x)$,& $(6_x, 4_x, 2_x, 1_x, 2_x, 4_x, 4_x, 6_x, 1_x)$,& $(6_x, 4_x, 2_x, 2_x, 1_x, 6_x, 1_x, 2_x, 4_x)$ \\ 
$(6_x, 4_x, 2_x, 2_x, 1_x, 6_x, 4_x, 6_x, 1_x)$,& $(6_x, 4_x, 2_x, 4_x, 6_x, 1_x, 1_x, 2_x, 4_x)$,& $(6_x, 4_x, 2_x, 4_x, 6_x, 1_x, 2_x, 1_x, 6_x)$ \\
\end{longtable}
\end{small}

\section{The $4\times 4$ orthogonal MDS matrices with optimal d-XOR count $72$ over $\mathbb{F}_{2^4}/0$x$13$}\label{Appendix:4-OMDS_F_2^4}

The tuples in the following are represented as $(m_{11}, m_{12}, m_{13}, m_{21}, m_{22}, m_{23}, m_{31}, m_{32}, m_{33})$, where $m_{ij}$ ($1 \leq i,j \leq 3$) are defined by equation (\ref{Eqn:GenForm-4Ortho}).

\begin{small}
\begin{longtable}{c c c}
$(1_x, 4_x, 9_x, 4_x, 1_x, d_x, 9_x, d_x, 1_x)$,& $(1_x, 4_x, 9_x, 4_x, 1_x, d_x, d_x, 9_x, 4_x)$,& $(1_x, 4_x, 9_x, 9_x, d_x, 1_x, 4_x, 1_x, d_x)$ \\ 
$(1_x, 4_x, 9_x, 9_x, d_x, 1_x, d_x, 9_x, 4_x)$,& $(1_x, 4_x, 9_x, d_x, 9_x, 4_x, 4_x, 1_x, d_x)$,& $(1_x, 4_x, 9_x, d_x, 9_x, 4_x, 9_x, d_x, 1_x)$ \\ 
$(1_x, 4_x, d_x, 4_x, 1_x, 9_x, 9_x, d_x, 4_x)$,& $(1_x, 4_x, d_x, 4_x, 1_x, 9_x, d_x, 9_x, 1_x)$,& $(1_x, 4_x, d_x, 9_x, d_x, 4_x, 4_x, 1_x, 9_x)$ \\ 
$(1_x, 4_x, d_x, 9_x, d_x, 4_x, d_x, 9_x, 1_x)$,& $(1_x, 4_x, d_x, d_x, 9_x, 1_x, 4_x, 1_x, 9_x)$,& $(1_x, 4_x, d_x, d_x, 9_x, 1_x, 9_x, d_x, 4_x)$ \\ 
$(1_x, 9_x, 4_x, 4_x, d_x, 1_x, 9_x, 1_x, d_x)$,& $(1_x, 9_x, 4_x, 4_x, d_x, 1_x, d_x, 4_x, 9_x)$,& $(1_x, 9_x, 4_x, 9_x, 1_x, d_x, 4_x, d_x, 1_x)$ \\ 
$(1_x, 9_x, 4_x, 9_x, 1_x, d_x, d_x, 4_x, 9_x)$,& $(1_x, 9_x, 4_x, d_x, 4_x, 9_x, 4_x, d_x, 1_x)$,& $(1_x, 9_x, 4_x, d_x, 4_x, 9_x, 9_x, 1_x, d_x)$ \\ 
$(1_x, 9_x, d_x, 4_x, d_x, 9_x, 9_x, 1_x, 4_x)$,& $(1_x, 9_x, d_x, 4_x, d_x, 9_x, d_x, 4_x, 1_x)$,& $(1_x, 9_x, d_x, 9_x, 1_x, 4_x, 4_x, d_x, 9_x)$ \\
$(1_x, 9_x, d_x, 9_x, 1_x, 4_x, d_x, 4_x, 1_x)$,& $(1_x, 9_x, d_x, d_x, 4_x, 1_x, 4_x, d_x, 9_x)$,& $(1_x, 9_x, d_x, d_x, 4_x, 1_x, 9_x, 1_x, 4_x)$ \\ 
$(1_x, d_x, 4_x, 4_x, 9_x, 1_x, 9_x, 4_x, d_x)$,& $(1_x, d_x, 4_x, 4_x, 9_x, 1_x, d_x, 1_x, 9_x)$,& $(1_x, d_x, 4_x, 9_x, 4_x, d_x, 4_x, 9_x, 1_x)$ \\ 
$(1_x, d_x, 4_x, 9_x, 4_x, d_x, d_x, 1_x, 9_x)$,& $(1_x, d_x, 4_x, d_x, 1_x, 9_x, 4_x, 9_x, 1_x)$,& $(1_x, d_x, 4_x, d_x, 1_x, 9_x, 9_x, 4_x, d_x)$ \\ 
$(1_x, d_x, 9_x, 4_x, 9_x, d_x, 9_x, 4_x, 1_x)$,& $(1_x, d_x, 9_x, 4_x, 9_x, d_x, d_x, 1_x, 4_x)$,& $(1_x, d_x, 9_x, 9_x, 4_x, 1_x, 4_x, 9_x, d_x)$ \\
$(1_x, d_x, 9_x, 9_x, 4_x, 1_x, d_x, 1_x, 4_x)$,& $(1_x, d_x, 9_x, d_x, 1_x, 4_x, 4_x, 9_x, d_x)$,& $(1_x, d_x, 9_x, d_x, 1_x, 4_x, 9_x, 4_x, 1_x)$ \\ 
$(4_x, 1_x, 9_x, 1_x, 4_x, d_x, 9_x, d_x, 4_x)$,& $(4_x, 1_x, 9_x, 1_x, 4_x, d_x, d_x, 9_x, 1_x)$,& $(4_x, 1_x, 9_x, 9_x, d_x, 4_x, 1_x, 4_x, d_x)$ \\ 
$(4_x, 1_x, 9_x, 9_x, d_x, 4_x, d_x, 9_x, 1_x)$,& $(4_x, 1_x, 9_x, d_x, 9_x, 1_x, 1_x, 4_x, d_x)$,& $(4_x, 1_x, 9_x, d_x, 9_x, 1_x, 9_x, d_x, 4_x)$ \\ 
$(4_x, 1_x, d_x, 1_x, 4_x, 9_x, 9_x, d_x, 1_x)$,& $(4_x, 1_x, d_x, 1_x, 4_x, 9_x, d_x, 9_x, 4_x)$,& $(4_x, 1_x, d_x, 9_x, d_x, 1_x, 1_x, 4_x, 9_x)$ \\ 
$(4_x, 1_x, d_x, 9_x, d_x, 1_x, d_x, 9_x, 4_x)$,& $(4_x, 1_x, d_x, d_x, 9_x, 4_x, 1_x, 4_x, 9_x)$,& $(4_x, 1_x, d_x, d_x, 9_x, 4_x, 9_x, d_x, 1_x)$ \\ 
$(4_x, 9_x, 1_x, 1_x, d_x, 4_x, 9_x, 4_x, d_x)$,& $(4_x, 9_x, 1_x, 1_x, d_x, 4_x, d_x, 1_x, 9_x)$,& $(4_x, 9_x, 1_x, 9_x, 4_x, d_x, 1_x, d_x, 4_x)$ \\ 
$(4_x, 9_x, 1_x, 9_x, 4_x, d_x, d_x, 1_x, 9_x)$,& $(4_x, 9_x, 1_x, d_x, 1_x, 9_x, 1_x, d_x, 4_x)$,& $(4_x, 9_x, 1_x, d_x, 1_x, 9_x, 9_x, 4_x, d_x)$ \\ 
$(4_x, 9_x, d_x, 1_x, d_x, 9_x, 9_x, 4_x, 1_x)$,& $(4_x, 9_x, d_x, 1_x, d_x, 9_x, d_x, 1_x, 4_x)$,& $(4_x, 9_x, d_x, 9_x, 4_x, 1_x, 1_x, d_x, 9_x)$ \\ 
$(4_x, 9_x, d_x, 9_x, 4_x, 1_x, d_x, 1_x, 4_x)$,& $(4_x, 9_x, d_x, d_x, 1_x, 4_x, 1_x, d_x, 9_x)$,& $(4_x, 9_x, d_x, d_x, 1_x, 4_x, 9_x, 4_x, 1_x)$ \\ 
$(4_x, d_x, 1_x, 1_x, 9_x, 4_x, 9_x, 1_x, d_x)$,& $(4_x, d_x, 1_x, 1_x, 9_x, 4_x, d_x, 4_x, 9_x)$,& $(4_x, d_x, 1_x, 9_x, 1_x, d_x, 1_x, 9_x, 4_x)$ \\ 
$(4_x, d_x, 1_x, 9_x, 1_x, d_x, d_x, 4_x, 9_x)$,& $(4_x, d_x, 1_x, d_x, 4_x, 9_x, 1_x, 9_x, 4_x)$,& $(4_x, d_x, 1_x, d_x, 4_x, 9_x, 9_x, 1_x, d_x)$ \\ 
$(4_x, d_x, 9_x, 1_x, 9_x, d_x, 9_x, 1_x, 4_x)$,& $(4_x, d_x, 9_x, 1_x, 9_x, d_x, d_x, 4_x, 1_x)$,& $(4_x, d_x, 9_x, 9_x, 1_x, 4_x, 1_x, 9_x, d_x)$ \\ 
$(4_x, d_x, 9_x, 9_x, 1_x, 4_x, d_x, 4_x, 1_x)$,& $(4_x, d_x, 9_x, d_x, 4_x, 1_x, 1_x, 9_x, d_x)$,& $(4_x, d_x, 9_x, d_x, 4_x, 1_x, 9_x, 1_x, 4_x)$ \\ 
$(9_x, 1_x, 4_x, 1_x, 9_x, d_x, 4_x, d_x, 9_x)$,& $(9_x, 1_x, 4_x, 1_x, 9_x, d_x, d_x, 4_x, 1_x)$,& $(9_x, 1_x, 4_x, 4_x, d_x, 9_x, 1_x, 9_x, d_x)$ \\ 
$(9_x, 1_x, 4_x, 4_x, d_x, 9_x, d_x, 4_x, 1_x)$,& $(9_x, 1_x, 4_x, d_x, 4_x, 1_x, 1_x, 9_x, d_x)$,& $(9_x, 1_x, 4_x, d_x, 4_x, 1_x, 4_x, d_x, 9_x)$ \\ 
$(9_x, 1_x, d_x, 1_x, 9_x, 4_x, 4_x, d_x, 1_x)$,& $(9_x, 1_x, d_x, 1_x, 9_x, 4_x, d_x, 4_x, 9_x)$,& $(9_x, 1_x, d_x, 4_x, d_x, 1_x, 1_x, 9_x, 4_x)$ \\ 
$(9_x, 1_x, d_x, 4_x, d_x, 1_x, d_x, 4_x, 9_x)$,& $(9_x, 1_x, d_x, d_x, 4_x, 9_x, 1_x, 9_x, 4_x)$,& $(9_x, 1_x, d_x, d_x, 4_x, 9_x, 4_x, d_x, 1_x)$ \\ 
$(9_x, 4_x, 1_x, 1_x, d_x, 9_x, 4_x, 9_x, d_x)$,& $(9_x, 4_x, 1_x, 1_x, d_x, 9_x, d_x, 1_x, 4_x)$,& $(9_x, 4_x, 1_x, 4_x, 9_x, d_x, 1_x, d_x, 9_x)$ \\ 
$(9_x, 4_x, 1_x, 4_x, 9_x, d_x, d_x, 1_x, 4_x)$,& $(9_x, 4_x, 1_x, d_x, 1_x, 4_x, 1_x, d_x, 9_x)$,& $(9_x, 4_x, 1_x, d_x, 1_x, 4_x, 4_x, 9_x, d_x)$ \\ 
$(9_x, 4_x, d_x, 1_x, d_x, 4_x, 4_x, 9_x, 1_x)$,& $(9_x, 4_x, d_x, 1_x, d_x, 4_x, d_x, 1_x, 9_x)$,& $(9_x, 4_x, d_x, 4_x, 9_x, 1_x, 1_x, d_x, 4_x)$ \\ 
$(9_x, 4_x, d_x, 4_x, 9_x, 1_x, d_x, 1_x, 9_x)$,& $(9_x, 4_x, d_x, d_x, 1_x, 9_x, 1_x, d_x, 4_x)$,& $(9_x, 4_x, d_x, d_x, 1_x, 9_x, 4_x, 9_x, 1_x)$ \\ 
$(9_x, d_x, 1_x, 1_x, 4_x, 9_x, 4_x, 1_x, d_x)$,& $(9_x, d_x, 1_x, 1_x, 4_x, 9_x, d_x, 9_x, 4_x)$,& $(9_x, d_x, 1_x, 4_x, 1_x, d_x, 1_x, 4_x, 9_x)$ \\ 
$(9_x, d_x, 1_x, 4_x, 1_x, d_x, d_x, 9_x, 4_x)$,& $(9_x, d_x, 1_x, d_x, 9_x, 4_x, 1_x, 4_x, 9_x)$,& $(9_x, d_x, 1_x, d_x, 9_x, 4_x, 4_x, 1_x, d_x)$ \\ 
$(9_x, d_x, 4_x, 1_x, 4_x, d_x, 4_x, 1_x, 9_x)$,& $(9_x, d_x, 4_x, 1_x, 4_x, d_x, d_x, 9_x, 1_x)$,& $(9_x, d_x, 4_x, 4_x, 1_x, 9_x, 1_x, 4_x, d_x)$ \\ 
$(9_x, d_x, 4_x, 4_x, 1_x, 9_x, d_x, 9_x, 1_x)$,& $(9_x, d_x, 4_x, d_x, 9_x, 1_x, 1_x, 4_x, d_x)$,& $(9_x, d_x, 4_x, d_x, 9_x, 1_x, 4_x, 1_x, 9_x)$ \\ 
$(d_x, 1_x, 4_x, 1_x, d_x, 9_x, 4_x, 9_x, d_x)$,& $(d_x, 1_x, 4_x, 1_x, d_x, 9_x, 9_x, 4_x, 1_x)$,& $(d_x, 1_x, 4_x, 4_x, 9_x, d_x, 1_x, d_x, 9_x)$ \\ 
$(d_x, 1_x, 4_x, 4_x, 9_x, d_x, 9_x, 4_x, 1_x)$,& $(d_x, 1_x, 4_x, 9_x, 4_x, 1_x, 1_x, d_x, 9_x)$,& $(d_x, 1_x, 4_x, 9_x, 4_x, 1_x, 4_x, 9_x, d_x)$ \\
$(d_x, 1_x, 9_x, 1_x, d_x, 4_x, 4_x, 9_x, 1_x)$,& $(d_x, 1_x, 9_x, 1_x, d_x, 4_x, 9_x, 4_x, d_x)$,& $(d_x, 1_x, 9_x, 4_x, 9_x, 1_x, 1_x, d_x, 4_x)$ \\ 
$(d_x, 1_x, 9_x, 4_x, 9_x, 1_x, 9_x, 4_x, d_x)$,& $(d_x, 1_x, 9_x, 9_x, 4_x, d_x, 1_x, d_x, 4_x)$,& $(d_x, 1_x, 9_x, 9_x, 4_x, d_x, 4_x, 9_x, 1_x)$ \\ 
$(d_x, 4_x, 1_x, 1_x, 9_x, d_x, 4_x, d_x, 9_x)$,& $(d_x, 4_x, 1_x, 1_x, 9_x, d_x, 9_x, 1_x, 4_x)$,& $(d_x, 4_x, 1_x, 4_x, d_x, 9_x, 1_x, 9_x, d_x)$ \\ 
$(d_x, 4_x, 1_x, 4_x, d_x, 9_x, 9_x, 1_x, 4_x)$,& $(d_x, 4_x, 1_x, 9_x, 1_x, 4_x, 1_x, 9_x, d_x)$,& $(d_x, 4_x, 1_x, 9_x, 1_x, 4_x, 4_x, d_x, 9_x)$ \\ 
$(d_x, 4_x, 9_x, 1_x, 9_x, 4_x, 4_x, d_x, 1_x)$,& $(d_x, 4_x, 9_x, 1_x, 9_x, 4_x, 9_x, 1_x, d_x)$,& $(d_x, 4_x, 9_x, 4_x, d_x, 1_x, 1_x, 9_x, 4_x)$ \\ 
$(d_x, 4_x, 9_x, 4_x, d_x, 1_x, 9_x, 1_x, d_x)$,& $(d_x, 4_x, 9_x, 9_x, 1_x, d_x, 1_x, 9_x, 4_x)$,& $(d_x, 4_x, 9_x, 9_x, 1_x, d_x, 4_x, d_x, 1_x)$ \\ 
$(d_x, 9_x, 1_x, 1_x, 4_x, d_x, 4_x, 1_x, 9_x)$,& $(d_x, 9_x, 1_x, 1_x, 4_x, d_x, 9_x, d_x, 4_x)$,& $(d_x, 9_x, 1_x, 4_x, 1_x, 9_x, 1_x, 4_x, d_x)$ \\ 
$(d_x, 9_x, 1_x, 4_x, 1_x, 9_x, 9_x, d_x, 4_x)$,& $(d_x, 9_x, 1_x, 9_x, d_x, 4_x, 1_x, 4_x, d_x)$,& $(d_x, 9_x, 1_x, 9_x, d_x, 4_x, 4_x, 1_x, 9_x)$ \\ 
$(d_x, 9_x, 4_x, 1_x, 4_x, 9_x, 4_x, 1_x, d_x)$,& $(d_x, 9_x, 4_x, 1_x, 4_x, 9_x, 9_x, d_x, 1_x)$,& $(d_x, 9_x, 4_x, 4_x, 1_x, d_x, 1_x, 4_x, 9_x)$ \\ 
$(d_x, 9_x, 4_x, 4_x, 1_x, d_x, 9_x, d_x, 1_x)$,& $(d_x, 9_x, 4_x, 9_x, d_x, 1_x, 1_x, 4_x, 9_x)$,& $(d_x, 9_x, 4_x, 9_x, d_x, 1_x, 4_x, 1_x, d_x)$  
\end{longtable}
\end{small}

\end{document}